\documentclass[11pt]{article}
\usepackage[utf8]{inputenc}
\usepackage[T1]{fontenc}
\usepackage[a4paper, margin=25mm]{geometry}
\usepackage{lmodern,microtype,algorithm,algpseudocode,amsthm,amsfonts,mathtools,thmtools,thm-restate,tikz,graphicx,hyperref}
\usepackage{enumitem}\setlist[enumerate,1]{label={(\alph*)}}

\usetikzlibrary{arrows.meta}

\theoremstyle{plain}
\declaretheorem[name=Theorem]{theorem}

\declaretheorem[sibling=theorem,name=Corollary]{corollary}

\declaretheorem[sibling=theorem,name=Conjecture]{conjecture}
\declaretheorem[name=Claim]{claim}

\title{Flood-It with Jewelry -- Characterizing the Game Complexity for Cograph Generalizations}

\author{Martin Darm\"{u}ntzel\\[0.5em]University of Rostock\\Germany\\[0.2cm]\texttt{\href{mailto:martin.darmuentzel@uni-rostock.de}{martin.darmuentzel@uni-rostock.de}}
	\and Christian Rosenke\\[0.5em]University of Rostock\\Germany\\[0.2cm]\texttt{\href{mailto:christian.rosenke@uni-rostock.de}{christian.rosenke@uni-rostock.de}}
	\and Mark Scheibner\\[0.5em]Independent Researcher\\Germany\\[0.2cm]\texttt{\href{mailto:mail@mark-scheibner.de}{mail@mark-scheibner.de}}}

\date{}

\usepackage{algorithm,algpseudocode,amsmath,amssymb,amsfonts,tikz}

\newcommand{\FloodIt}{Flood-It}
\newcommand{\FreeFloodIt}{Free Flood-It}
\newcommand{\MinimumFloodIt}{\textsc{Flood-It}}
\newcommand{\MinimumFreeFloodIt}{\textsc{Free Flood-It}}
\newcommand{\FIt}{FIt}
\newcommand{\FFIt}{FFIt}

\newcommand{\cG}{\mathcal{G}}

\newcommand{\floodarea}{flooded area}
\newcommand{\floodareas}{flooded areas}

\newcommand{\flooding}{flooding}
\newcommand{\floodings}{floodings}
\newcommand{\playboard}{playfield}

\newcommand{\Playboards}{Playfields}

\newcommand{\gemlets}{jewels}

\newcommand{\Gemlets}{Jewels}

\newcommand{\NP}{\mathsf{NP}}

\newcommand{\pfour}{\ensuremath{P_4}}
\newcommand{\pfive}{\ensuremath{P_5}}
\newcommand{\house}{\text{house}} 
\newcommand{\chair}{\text{chair}}
\newcommand{\kite}{\text{kite}} 
\newcommand{\gem}{\text{gem}}
\newcommand{\cogem}{\text{co-gem}} 
\newcommand{\banner}{\text{banner}}
\newcommand{\cobanner}{\text{co-banner}}
\newcommand{\bull}{\text{bull}} 
\newcommand{\cfive}{\ensuremath{C_5}} 

\newcommand{\claw}{\text{claw}}
\newcommand{\dimond}{\text{diamond}}
\newcommand{\twoKtwo}{\ensuremath{2K_2}}
\newcommand{\cfour}{\ensuremath{C_4}}

\newcommand{\cO}{{\cal O}}
\newcommand{\N}{\mathbb{N}}

\newcommand{\tip}{tip}
\newcommand{\dist}{\ensuremath{\operatorname{dist}}}

\begin{document}

\maketitle

\begin{abstract}
	\emph{\FloodIt} is a single-player game played on a precolored graph $G$, where the objective is to make $G$ monochromatic using as few \emph{flooding} moves as possible.
	In each move, a color $c$ is selected and all vertices reachable from a fixed pivot vertex via a monochromatic path are recolored with $c$.
	In the \emph{free} variant, the pivot may be chosen anew in every move.

	Deciding whether a graph can be made monochromatic in at most $k$ moves is $\NP$-complete for both variants, fixed and free.
	This hardness persists even under strong structural restrictions such as split graphs and trees.
	The \emph{\FreeFloodIt} variant is generally considered more difficult than its fixed-pivot counterpart, as it remains hard on several graph classes where the latter becomes tractable, including co-comparability and AT-free graphs.

	Cographs, that is, {\pfour}-free graphs, are among the few classes on which even {\FreeFloodIt} is solvable in polynomial time and therefore serve as our starting point.
	We consider the ten natural one-vertex extensions of {\pfour} -- referred to as \emph{\gemlets} -- and study the complexity of both flooding games on the $1024$ graph classes obtained by forbidding subsets of these graphs as induced subgraphs.

	Our main contribution is a polynomial-time algorithm for {\FreeFloodIt} on graphs that are free of the three {\gemlets} {\bull}, {\gem}, and \pfive, covering $128$ of the $1024$ classes.
	In addition, we prove that both variants remain $\NP$-complete on thin-spider graphs, which exclude the eight {\gemlets} {\banner}, {\cobanner}, {\chair}, {\gem}, {\house}, {\kite}, {\pfive}, and {\cfive}, thereby establishing hardness for $256$ additional classes.

	Combined with known algorithms and hardness results, our work determines the complexity of both {\FloodIt} variants for $896$ of the $1024$ considered graph classes.
\end{abstract}
\newpage


\section{Introduction}

\emph{Flood-It} and several of its variants have been studied extensively between 2010 and 2020.
This substantial body of work has shown that computing optimal strategies remains computationally hard, even under severe restrictions on the input.
For comprehensive overviews of known hardness results and the comparatively few polynomial-time solvable cases, we refer the interested reader to the excellent surveys by Fellows et al.~\cite{Fellows2018} and Fleischer and Woeginger~\cite{Fleischer2010}.

While the original flood-filling game is played on precolored grids, the notion of a {\playboard} has been generalized in the literature to arbitrary graphs to better capture and analyze the underlying algorithmic challenges.
In this paper, we focus on single-player variants, where each game is described by a \emph{move order}, that is, a sequence of moves, each of which floods a chosen color $c$ onto all vertices $w$ that are reachable from a pivot vertex $v$ via a monochromatic path, that is, a path from $v$ to $w$ whose vertices all have the same color at the time of the move.
The goal of the game is to make the {\playboard} graph $G$ monochromatic; accordingly, any move order that achieves this is called a \emph{strategy} for $G$.
In the variant called \emph{\FreeFloodIt} the player may \emph{freely} choose the pivot vertex in every move.
Naturally, if the pivot is a \emph{fixed} vertex specified as part of the input, we obtain the \emph{(fixed) \FloodIt} game.
In both variants, the objective is to find a strategy of minimum length.
There are corresponding decision variants of these problems that additionally take an integer $k$ as input and ask whether there exists a strategy using at most $k$ moves.
Since the relevant variant, optimization or decision, will always be clear from context, this paper does not introduce separate notations.

It is known that the computational problems underlying these flood-filling games, {\MinimumFloodIt} and {\MinimumFreeFloodIt}, remain intractable even under severe input restrictions such as \emph{split graphs}~\cite{Fleischer2010, Fukui2013}.
In particular, both problems are $\NP$-hard when the input {\playboard} is a \emph{tree}~\cite{FELLOWS2015} or a \emph{grid}~\cite{Clifford2012}, even if it is initially colored with only three colors.

Only very few non-trivial graph classes are known on which {\MinimumFloodIt} ({\FIt}, for short) is tractable, including \emph{co-comparability}~\cite{Fleischer2010} and \emph{AT-free} graphs~\cite{Fukui2013}.
In contrast, {\MinimumFreeFloodIt} ({\FFIt}, for short) remains hard even on these classes.
For many years, polynomial-time solvability of the free variant has been known only for very restricted grids~\cite{MEEKS2012}, as well as for simple paths and cycles~\cite{MeeksScott2014}.
However, for the second powers of paths and cycles~\cite{souza2014}, {\FFIt} becomes $\NP$-hard again.

As part of the algorithmic techniques developed in this paper, we show that deciding whether a graph $G$ can be made monochromatic using exactly $k-1$ free moves -- one for each of the $k$ colors used in $G$ except the last -- can be done in polynomial time, independent of the presence of certain structures in the graph and the initial coloring.
Indeed, if the {\playboard} initially uses $k$ distinct colors, any strategy requires at least $k-1$ moves, since each move can eliminate at most one color.
Given the apparent resistance of {\FFIt} to efficient algorithms, it is somewhat surprising that deciding the existence of a $k-1$-move strategy is always in polynomial time.

The main contribution of this paper, however, is motivated by the recent result~\cite{rosenke2026} that {\FFIt} is polynomial-time solvable on \emph{cographs} -- the graphs without induced {\pfour}.
This naturally raises the question of how the restriction of forbidding induced {\pfour} can be relaxed without losing tractability.
In~\cite{rosenke2026}, this question is addressed by considering the extension of {\pfour} by an independent vertex and deriving a significantly more involved polynomial-time algorithm for graphs that are free of induced {\cogem}s; see Figure~\ref{fig:example}~(f).

It is well known that the {\cogem} is not the only graph generalizing {\pfour}, and that there are exactly ten non-isomorphic one-vertex extensions of {\pfour}; see, for example,~\cite{brandstaedt2005}.
As illustrated in Figure~\ref{fig:example}, these graphs can be viewed as interpolating between the {\cogem} and its complement, the {\gem}, by selecting a subset of the four edges that connect the additional vertex to the underlying {\pfour}.
Motivated by this interpretation, we refer to these graphs collectively as \emph{\gemlets}.

\begin{figure}[h!]
  \def\myscale{0.95}
  \tikzstyle{vertex}=[draw,circle,fill=black,inner sep=1pt]
  \tikzstyle{ivertex}=[draw,circle,inner sep=.8pt,fill=white]
  \tikzstyle{edges}=[draw=black!66, line join=round]
  \newcommand{\drawGemlet}[4]{
    \draw[edges] (-0.99,0) node[vertex] (v1) {} \foreach \i in {2,3,4}{ -- ++(0.66,0) node[vertex] (v\i) {}};
	\node[#4] (x) at (0,-1) {};
	\foreach \i in {#1}{\draw[edges] (x) -- (v\i);}
  	\node[below=0.5cm] at (x) {(#3)};
	\path (v1) -- (v4) node[pos=0.5, above=0.33cm] {#2};
  }
  \newcommand{\drawSmallGraph}[2]{
	\path (-0.33,0) node[vertex] (v1) {} \foreach \i/\dx/\dy in {2/0.66/0,3/0/-0.66,4/-0.66/0}{ -- ++(\dx,\dy) node[vertex] (v\i) {}};
	\foreach \v/\w in {#1} {\draw[edges] (v\v) -- (v\w);}
	\coordinate (x) at (0,-1);
	\path (v1) -- (v2) node[pos=0.5, above=0.33cm] {#2};
  }

  \begin{tikzpicture}[scale = \myscale, every node/.style={scale=\myscale}]
	\begin{scope}[xshift=0cm]
    	\drawGemlet{1,2,3,4}{\gem}{a}{vertex}
	\end{scope}
	\begin{scope}[xshift=3cm]
    	\drawGemlet{1}{\pfive}{b}{vertex}
	\end{scope}
	\begin{scope}[xshift=6cm]
    	\drawGemlet{2}{\chair}{c}{vertex}
	\end{scope}
	\begin{scope}[xshift=9cm]
    	\drawGemlet{1,3}{\banner}{d}{vertex}
	\end{scope}
	\begin{scope}[xshift=12cm]
    	\drawGemlet{1,4}{{\cfive} (also co-\cfive)}{e}{vertex}
	\end{scope}
  \end{tikzpicture}

  \vspace*{0.5cm}
  \begin{tikzpicture}[scale = \myscale, every node/.style={scale=\myscale}]
	\begin{scope}[xshift=0cm]
    	\drawGemlet{}{\cogem}{f}{vertex}
	\end{scope}
	\begin{scope}[xshift=3cm]
    	\drawGemlet{1,2,4}{{\house} (co-\pfive)}{g}{vertex}
	\end{scope}
	\begin{scope}[xshift=6cm]
    	\drawGemlet{1,2,3}{{\kite} (co-\chair)}{h}{vertex}
	\end{scope}
	\begin{scope}[xshift=9cm]
    	\drawGemlet{1,2}{\cobanner}{i}{vertex}
	\end{scope}
	\begin{scope}[xshift=12cm]
    	\drawGemlet{2,3}{{\bull} (also co-\bull)}{j}{vertex}
	\end{scope}
  \end{tikzpicture}

  \caption{The {\gemlets}, one-vertex extensions of the {\pfour}, given with common name and complement graph relation}
  \label{fig:example}
\end{figure}
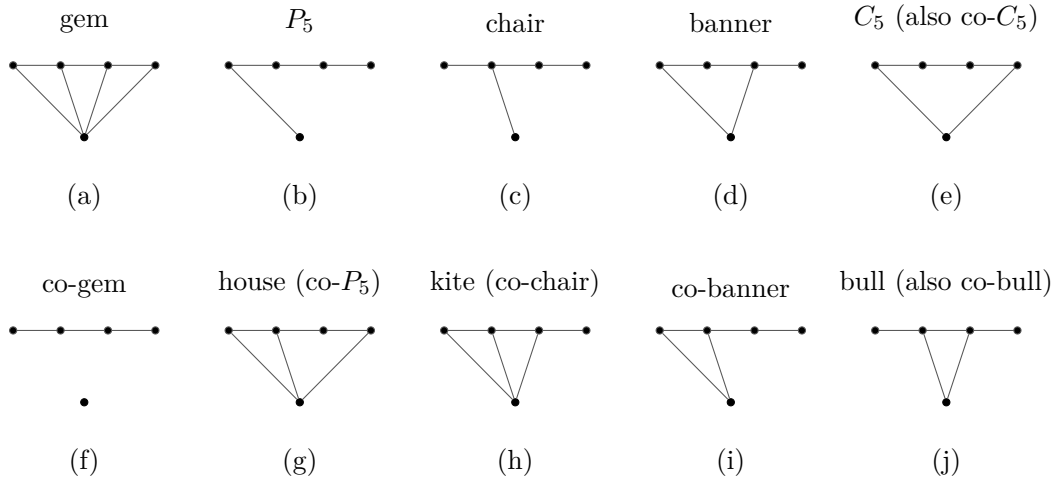

Now, if we consider arbitrary subsets of these small graphs as forbidden induced subgraphs, we obtain a large collection of natural generalizations of cographs.
Based on this, we continue the line of inquiry initiated in~\cite{rosenke2026} in a natural way and ask, for every possible subset $\cG$ of the {\gemlets}, what the complexity of (Free) {\FloodIt} is when the input is restricted to the graph class defined by forbidding all graphs in $\cG$.
Since there are ten {\gemlets}, this framework gives rise to a family of $2^{10}=1024$ graph classes.
The empty set of forbidden graphs yields the class of all graphs, whereas forbidding all ten {\gemlets} results in the class of cographs (together with the single graph~$\pfour$).

\medskip
Given the size of this family, it is natural to revisit existing results in order to preemptively resolve as many cases as possible.
Starting with the tractable cases, and noting that adding further forbidden induced subgraphs cannot destroy tractability, the recent result~\cite{rosenke2026} already implies that {\FreeFloodIt} is solvable in polynomial time on the $512$ classes that include {\cogem} in the set of forbidden graphs.
\begin{corollary}\label{crl:P-Time:from_cogemfree}
	{\MinimumFreeFloodIt} can be solved in polynomial time on every graph class whose set of forbidden induced {\gemlets} contains {\cogem}.
\end{corollary}
A closer look at~\cite{rosenke2026} reveals that the paper studies the related \textsc{Miniature Painting} problem and polynomial-time solvability of {\FreeFloodIt} follows only as a corollary.
This implication does not readily extend to the fixed-pivot variant.
Nevertheless, no graph class is currently known on which the free variant is tractable while the fixed variant remains hard.
This motivates the conjecture that {\FloodIt} is also polynomial-time solvable on \cogem-free graphs.

The main result of the present paper follows the general direction of the previous work in~\cite{rosenke2026} and develops new algorithmic techniques for graph classes defined by forbidding {\pfive}, {\bull}, and {\gem}.
Specifically, Section~\ref{sec:graph_structure} analyzes the structure of $(\pfive,\bull,\gem)$-free graphs and derives properties that enable efficient algorithms.
These insights are then used in Sections~\ref{sec:Polytime_Edge_Case} and~\ref{sec:Polytime_Bead_Bracelets} to obtain a polynomial-time algorithm for {\FreeFloodIt} on this class.
This leads to the following theorem.
\begin{theorem}\label{thm:polynomial_time+cases}
	{\MinimumFreeFloodIt} can be solved in polynomial time on every graph class whose forbidden induced {\gemlets} include {\pfive}, {\bull}, and {\gem}.
\end{theorem}
Our result is stated only for {\FFIt}, as the polynomial-time algorithm for AT-free graphs from~\cite{atfree} already implies tractability of the fixed-pivot variant on $(\pfive,\bull,\gem)$-free graphs, which are AT-free.
{\FFIt}, however, remains hard on AT-free graphs~\cite{Fukui2013}; thus, this paper establishes a boundary between tractability and hardness within this graph class.
The main result covers $128$ of the $1024$ classes of cograph generalizations.
However, beyond the $512$ classes already addressed, only $64$ are new.

\medskip
To identify classes with intractable flooding strategies, we take a closer look at the $\NP$-completeness proofs by Fleischer and Woeginger~\cite{Fleischer2010} and by Fukui et~al.~\cite{Fukui2013}.
These reveal that their reductions for split graphs satisfy strong structural constraints.
In general, split graphs are free of induced {\pfive}; see Figure~\ref{fig:example}~(b).
However, the constructed gadget graphs are, in fact, also free of {\house}, {\banner}, {\cobanner}, {\cfive}, {\kite}, and {\gem}.
Consequently, both flooding games remain $\NP$-complete on every graph class defined by forbidding an arbitrary subset of these {\gemlets}.

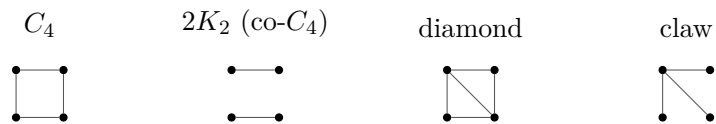
\begin{figure}[b!]
  \def\myscale{0.95}
  \tikzstyle{vertex}=[draw,circle,fill=black,inner sep=1pt]
  \tikzstyle{ivertex}=[draw,circle,inner sep=.8pt,fill=white]
  \tikzstyle{edges}=[draw=black!66, line join=round]
  \newcommand{\drawSmallGraph}[2]{
	\path (-0.33,0) node[vertex] (v1) {} \foreach \i/\dx/\dy in {2/0.66/0,3/0/-0.66,4/-0.66/0}{ -- ++(\dx,\dy) node[vertex] (v\i) {}};
	\foreach \v/\w in {#1} {\draw[edges] (v\v) -- (v\w);}
	\coordinate (x) at (0,-1);
	\path (v1) -- (v2) node[pos=0.5, above=0.33cm] {#2};
  }
  \vspace*{0.5cm}
  \centering
  \begin{tikzpicture}[scale = \myscale, every node/.style={scale=\myscale}]
	\begin{scope}[xshift=3cm]
		\drawSmallGraph{1/2,2/3,3/4,4/1}{\cfour}
	\end{scope}
	\begin{scope}[xshift=6cm]
		\drawSmallGraph{1/2,3/4}{{\twoKtwo} (co-\cfour)}
	\end{scope}
	\begin{scope}[xshift=9cm]
		\drawSmallGraph{1/2,2/3,3/4,4/1,1/3}{\dimond}
	\end{scope}
	\begin{scope}[xshift=12cm]
		\drawSmallGraph{1/2,1/3,1/4}{\claw}
	\end{scope}
  \end{tikzpicture}

  \caption{Additional graphs appearing in this paper}
  \label{fig:smallgraphs}
\end{figure}

As a second contribution of this paper, Section~\ref{sec:np_completeness} refines these proof techniques by employing \emph{thin spiders} as reduction gadgets.
These split graphs are, in addition to the seven {\gemlets} listed above, also free of induced {\dimond} and {\claw}; see Figure~\ref{fig:smallgraphs}.
Since {\chair} is a natural generalization of {\claw}, it can be added to the list of forbidden graphs.
As a result, Section~\ref{sec:np_completeness} establishes $\NP$-completeness for both variants of {\FloodIt} on a broad range of classes, summarized in the following theorem.
\begin{theorem}\label{thm:NP-completeness:from_split_graphs}
	{\MinimumFloodIt} and {\MinimumFreeFloodIt} are $\NP$-complete on the class of $X$-free graphs for every $X \subseteq \{\pfive, \house, \chair, \kite, \banner, \cobanner, \cfive, \gem\}$.
\end{theorem}
Since this covers all subsets of eight of the ten {\gemlets}, the theorem resolves the complexity question for another $256$ classes in our family.
Implicitly, $128$ of them have been known before from the results of~\cite{Fleischer2010,Fukui2013}.

Another source of hardness is provided by a result of Fellows et~al.~\cite{FELLOWS2015}, who show that both game variants are $\NP$-complete when restricted to trees.
Because none of the cyclic {\gemlets} can appear as induced subgraphs of trees, forbidding them does not restrict this class.
Hence, their result directly yields the following corollary.
\begin{corollary}\label{crl:NP-completeness:from_trees}
	{\MinimumFloodIt} and {\MinimumFreeFloodIt} are $\NP$-complete on the class of $X$-free graphs for every $X \subseteq \{\banner, \cobanner, \bull, \cfive, \kite, \house, \gem\}$.
\end{corollary}
Among the $128$ classes covered by this corollary, those whose forbidden set contains {\bull} are not addressed by Theorem~\ref{thm:NP-completeness:from_split_graphs}.
This accounts for $64$ additional resolved cases.

\medskip
At this stage, $896$ of the $1024$ graph classes can be considered settled.
The results in this paper account for $192$ previously unresolved classes, representing substantial progress.
The remaining $128$ graph classes are yet to be classified.
We suspect that resolving them will require fundamentally different approaches, both in terms of reduction gadgets for NP-completeness proofs and structural insights for polynomial-time algorithms.

\medskip
After the next section introduces basic concepts and notation, the remainder of the paper is organized as follows.
Section~\ref{sec:graph_structure} develops structural properties of $(\pfive,\bull,\gem)$-free graphs, which are subsequently exploited in Sections~\ref{sec:Polytime_Edge_Case} and~\ref{sec:Polytime_Bead_Bracelets} to obtain a polynomial-time algorithm for {\FFIt} on this graph class.
In Section~\ref{sec:np_completeness}, we prove that both flood-filling games are $\NP$-complete on thin spiders and derive Theorem~\ref{thm:NP-completeness:from_split_graphs} as a direct consequence.
Finally, Section~\ref{sec:conclusion} concludes the paper.
To keep the presentation accesible, all proofs are deferred to Section~\ref{sec:proofs}.


\section{Preliminaries}
\label{sec:preliminaries}

Before presenting our results, we fix the notation used throughout the paper.
We mostly rely on standard mathematical and graph-theoretic notation.
For sets $A$ and $B$, we denote the \emph{union}, \emph{difference}, and \emph{intersection} by $A+B$, $A-B$, and $A \cap B$, respectively.
Singleton sets $\{x\}$ are identified with their only element $x$, whenever this is clear from the context.

\subsection*{Graph Notation}

All graphs considered in this paper are finite, undirected, and without multiple edges or loops.
Figure~\ref{fig:example} depict the \emph{\gemlets}, that is, all graphs that can be obtained by extending a {\pfour} with one additional vertex, called the \emph{\tip}.
As explained in the introduction, these graphs form the basis for a spectrum of graph classes, each representing a certain generalization of cographs.
Moreover, Figure~\ref{fig:smallgraphs} shows some small graphs that are also referenced in this paper.

For every vertex $x \in V(G)$, we define the \emph{open neighborhood} $N_G(x) = \{y \mid xy \in E(G)\}$ and the \emph{closed neighborhood} $N_G[x] = N_G(x)+x$.
Since, throughout the paper, the graph $G$ is always clear from the context, we omit it as a subscript and simply write $N(x)$ and $N[x]$.
We say that a vertex $x$ \emph{sees} a vertex $y$, if $y \in N(x)$.
If $|N(x)| = 1$ then $x$ is called a \emph{leaf} and the only vertex in $N(x)$ is the \emph{parent} of $x$.
For a subset $V \subseteq V(G)$, we define $N[V] = \bigcup_{v \in V} N[v]$.
A \emph{dominating set} $D$ of a graph $G$ is a subset of $V(G)$ such that $N[D] = V(G)$.

If $G$ is a graph, every subgraph $H=(V,E)$ of $G$ is specified by a vertex subset $V \subseteq V(G)$ and an edge subset $E \subseteq E(G)$ that contains only edges with both endpoints in $V$.
A subgraph $H$ is \emph{induced}, if $E$ contains exactly all edges $vw \in E(G)$ with $v,w \in V$.
For any vertex subset $V \subseteq V(G)$, we write $G[V]$ for the induced subgraph of $G$ on $V$.
We use $G-V$ as a shorthand for $G[V(G)-V]$.
Graphs $G$ and $H$ are \emph{isomorphic}, if there exists a bijection $\sigma \colon V(G) \rightarrow V(H)$ such that $xy \in E(G)$ if and only if $\sigma(x)\sigma(y) \in E(H)$.
If no induced subgraph of $G$ is isomorphic to a graph $H$, then $G$ is called $H$-\emph{free}.
A subgraph $H$ of $G$, induced or not, is said to \emph{dominate} $G$, if $V(H)$ is a dominating set of $G$.

For graphs $G$ and $H$ with disjoint vertex sets and a vertex $v \in V(G)$, the \emph{substitution} $G[v \leftarrow H]$ of $v$ with $H$ is the graph with vertex set $V(H) + V(G) - v$ and edge set $E(H) + E(G) - \{vw \mid w \in N_G(v)\}$,
together with all edges $uw$ for $u \in V(H)$ and $w \in N_G(v)$.
In this way, the graph $H$ replaces the single vertex $v$ in $G$ and becomes a \emph{module} of $G[v \leftarrow H]$, that is, a subgraph whose vertices are either all seen or all unseen by every vertex in $V(G)-v$.
Since the order of substituting graphs $H_1, \dots, H_k$ for pairwise distinct vertices $v_1, \dots, v_k$ of a graph $G$ is irrelevant, we write
$G[v_1, \dots, v_k \leftarrow H_1, \dots, H_k]$ to denote the result of these substitutions in arbitrary order.

\subsection*{Flood Filling for Graphs}

For clarity, we refer to any surjective function $f\colon V(G) \rightarrow C$ with color set $C$ as a \emph{\flooding} of a graph $G$, rather than using the more common term \emph{coloring}.
For every color $c \in C$, we write $f^{-1}(c) = \{v \in V(G) \mid f(v) = c\}$ for the set of vertices with color $c$.
Given a {\flooding} $f$ of $G$, a vertex subset $M \subseteq V$ is \emph{monochromatic}, if $M \subseteq f^{-1}(c)$ for some color $c$.
A monochromatic set $M$ is called a \emph{color component} of $G$, if $M$ is connected in $G$ and no proper superset of $M$ is both monochromatic and connected.
A color $c \in C$ is said to be \emph{connected} with respect to $f$, if $f^{-1}(c)$ forms a single color component.
Otherwise, it is \emph{disconnected}.
Moreover, for every vertex $v \in V(G)$, we denote by $f^{-1}(v)$ the color component containing $v$, which is a subset of $f^{-1}(c)$ for the color $c = f(v)$.
We also call $f^{-1}(v)$ the \emph{\floodarea} of $v$.

With these notions in place, we formally define the flood filling problems introduced earlier.
Let $G$ be a graph with initial {\flooding} $f_0\colon V(G) \rightarrow C$.
A (possibly empty) \emph{free move order} $(v_1,c_1), \dots, (v_s,c_s)$ consists of moves $(v_i,c_i)$, where each $v_i \in V(G)$ is a \emph{pivot vertex} and each $c_i \in C$ is a \emph{flooding color}.
Since free move orders are the main subject of this paper, we often just call them \emph{move order}.
The result of a move $(v_i,c_i)$ is a new {\flooding} $f_i\colon V(G) \rightarrow C$ defined by
\[f_i(v) = \begin{cases}
	c_i, & \text{if } v \in f_{i-1}^{-1}(v_i),\\
	f_{i-1}(v), &\text{otherwise.}
\end{cases}\]
In contrast to free move orders, a \emph{fixed move order} with pivot $p$ satisfies $p = v_1 = \dots = v_s$.
Because the pivot does not change, fixed move orders can be specified simply by the sequence of colors $c_1, \dots, c_s$.

The purpose of a move $(v_i,c_i)$ is typically to change the color of all vertices in the color component $f_{i-1}^{-1}(v_i)$ to $c_i$ so that, with respect to $f_i$, this component merges with parts of $f_{i-1}^{-1}(c_i)$, that is, with vertices that are already colored $c_i$.
To formalize this growth in \emph{territory}, we say that the move $(v_i,c_i)$ \emph{conquers} the vertices $f_i^{-1}(v_i) - f_{i-1}^{-1}(v_i)$ from the pivot $v_i$.
This notion is particularly intuitive in the fixed flooding setting, where all moves conquer territory from the same pivot.

A move order, fixed or free, is called a \emph{strategy} for $G$, $f_0$, and, in the fixed case, $p$, if the resulting {\flooding} $f_s$ is constant.
The {\MinimumFloodIt} problem asks, given $G$, $f_0$, and a pivot $p$, for a shortest fixed strategy, that is, one using the minimum possible number of moves. Analogously, {\MinimumFreeFloodIt} asks for a shortest free strategy.


\section{Restrictions for Graphs free of the {\Gemlets} {\pfive}, {\bull} and {\gem}}
\label{sec:graph_structure}

This section prepares the main result of the paper, which concerns the classification of a large family of cograph generalizations that admit polynomial-time algorithms.
As outlined in the introduction, a substantial portion of these classes has already been identified in the preceding work~\cite{rosenke2026} on the closely related \textsc{Miniature Painting} problem.
A central insight behind the previous approach is that the considered \cogem-free graphs cannot expand arbitrarily and that their vertices are, in a certain sense, restricted to the space around small subgraphs.
More precisely, if $\pi$ is an induced {\pfour} in a \cogem-free graph $G$, then $\pi$ forms a dominating set and thus serves as a kind of \emph{hub}.
From the perspective of {\FreeFloodIt}, this means that, since all vertices of $G$ lie in the neighborhood of $\pi$, a meaningful strategy may first conquer $\pi$ and then flood each remaining color once in order to make $G$ monochromatic.
In fact, we essentially show that there exists a shortest strategy for $G$ that conquers an induced {\pfour}, say $\pi$, within at most twelve moves and subsequently uses $\pi$ as a hub for all remaining moves.
Such strategies can then be brute-forced, yielding a polynomial-time algorithm.

However, while the existence of small, connected dominating sets in \cogem-free graphs provides a powerful starting point, it is by no means a turnkey solution for flood filling problems.
Considerable effort is still required to exploit these structural properties algorithmically.
Accordingly, the aim of the present section is to take a first step toward a similar understanding for a different collection of jewel-free graphs.
Specifically, we establish analogous conditions for $(\pfive, \bull, \gem)$-free graphs, identifying a new type of hub: a localized graph structure that, in a suitable sense, restricts all vertices of the graph to the space around it.

As every $(\pfive, \bull, \gem)$-free graph $G$ can be viewed as a generalized cograph, our overarching objective is to recover the remnants of cograph structure that persist within $G$.
In fact, we will show that, in most cases, $G$ decomposes into a cyclic arrangement of cographs.
To stay within our notional scheme and since cographs are, in the strict sense, free of jewels, we refer to such a structure as a \emph{bead bracelet} with the individual cographs being just \emph{beads}.
More precisely, for a cycle $C_k$ with $k \geq 3$, vertices $v_0, v_1, \dots, v_{k-1}$, and non-empty cographs $H_0, H_1, \dots, H_{k-1}$, the graph $C_k[v_0, v_1, \dots, v_{k-1} \leftarrow H_0, H_1, \dots, H_{k-1}]$ that substitutes the cycle vertices with the beads $H_0, H_1, \dots, H_{k-1}$ is called a \emph{$k$-bead bracelet}.
The following lemma formalizes the connection between this concept and $(\pfive, \bull, \gem)$-free graphs.
\begin{restatable}{lemma}{CFiveCographInjection}
\label{lem:structure:pfive_bull_gem_free:cfive_substitution}
	If $G$ is a connected $(\pfive, \bull, \gem)$-free graph that contains an induced {\cfive}, then $G$ is a $5$-bead bracelet.
\end{restatable}

The underlying {\cfive} keeps the surrounding cographs in proximity and therefore serves as an effective flooding hub.
Building on this observation, Section~\ref{sec:Polytime_Bead_Bracelets} develops a new polynomial-time approach that exploits this hub structure.
On the other hand, there also exist $(\pfive, \bull, \gem, \cfive)$-free graphs $G$.
Such graphs behave even more like cographs.
In particular, connected cographs are typically rich in dominating edges.
This property carries over to the present generalization and, consequently, every such graph contains at least one dominating edge.
\begin{restatable}{lemma}{EdgeCase}
\label{prp:structure:pfive_bull_gem_free:dominating_pfour}
	If $G$ is a connected $(\pfive, \bull, \gem, \cfive)$-free graph of at least two vertices then $G$ contains a dominating edge.
\end{restatable}

The situation described by the lemma is, quite literally, an \emph{edge case} and therefore calls for a separate treatment, which is carried out in the following section.


\section{The Edge Case}
\label{sec:Polytime_Edge_Case}

Trivially, free flooding in the \emph{edge}-case has a lower bound of $|C|-1$ moves.
Recall that $C$ is the image of the flooding $f_0$, or, more intuitively, the initial set of colors in $G$.
Moreover, an upper bound of $|C|$ moves follows immediately, since one can always conquer a dominating edge $e$ in a single move and then, using a pivot in $e$, flood all colors of $C$ but one to make the graph monochromatic.
As these are the only two possible cases, it suffices to decide whether a given graph $G$ under an initial {\flooding} $f_0$ admits a free strategy of $|C|-1$ moves.
If this is not the case, we can compute a dominating edge $e$ in polynomial time and fall back on the $|C|$-move strategy that exploits $e$.

Accordingly, this section focuses on a polynomial-time approach to determine a free strategy $S$ with $|C|-1$ moves, if one exists.
This approach, surprisingly, requires no restrictions on the given graph.
To achieve this minimum number of moves, the strategy $S$ must \emph{eliminate} exactly one color $c$ in each move.
More precisely, this means that while the color $c$ appears on some vertices before a move $i$, no vertex is colored with $c$ after that move.
Clearly, eliminating a color $c$ requires that the $c$-colored vertices form a single color component $f^{-1}_{i-1}(c)$ before move $i$, and that move $i$ uses a pivot vertex $v_i \in f^{-1}_{i-1}(c)$.

Consequently, every move of the strategy $S$ floods a color $c$ from a pivot in the color component $f^{-1}_{i-1}(c')$ of a connected color $c' \not= c$ and, in doing so, may create the now connected color $c$.
Our key observation is that, as long as this principle is respected, the precise ordering of moves is of secondary importance.
More specifically, we argue that a free strategy of this type can be obtained by repeatedly selecting a suitable disconnected color $c$, determining which connected colors $c'$ must be eliminated to make $c$ connected, and then performing the corresponding moves in arbitrary order.
Furthermore, this selection process can be done greedily.
Any disconnected color whose vertices are only separated by connected colors can be made connected through these moves, immediately.

To formalize this idea in Algorithm~\ref{alg:minimal_strategy}, we introduce the following notions used there.
A set $\Gamma \subseteq C$ of colors is said to \emph{connect} a color $c \in C$, if all pairs $v, w \in f_0^{-1}(c)$ are connected in the graph
$G\big[f_0^{-1}(c) + \bigcup_{\gamma \in \Gamma} f_0^{-1}(\gamma)\big]$.
Note that this induced subgraph does not need to be connected.
If it is, however, and $c \notin \Gamma$, we call the set $\{ f_0^{-1}(\gamma) \mid \gamma \in \Gamma \}$ a \emph{color region of $c$}.
A move that connects a previously disconnected color $c$ is called the \emph{connecting move} for $c$.
Accordingly, a move order that contains a connecting move for $c$ is also said to \emph{connect} $c$.
\begin{algorithm}
	\caption{Polynomial-time algorithm that returns a free strategy of $|C|-1$ moves, if one exists, and otherwise $\bot$.}
	\label{alg:minimal_strategy}
	\begin{algorithmic}[1]
		\Procedure{minimal\_strategy}{$G, f_0$}
			\State $S \gets $ empty move order
			\State $C' \gets C$ \Comment{Colors that have not been eliminated}
			\State $\Gamma \gets \{ c \in C \mid \text{c is connected} \}$ \Comment{Colors that have successfully been connected}
			\While{$\Gamma \neq C$}
				\State $\Delta \gets \{ c \in C \mid \Gamma \text{ connects } c\} - \Gamma$ \Comment{Colors that can be connected optimally}
				\If{$\Delta = \emptyset$}
					\State \Return $\bot$
				\EndIf
				\For{$c \in \Delta$}
					\State $P \gets $ a color region of $c$
					\State $S \gets S$ followed by a move $(v_R, c)$ for each $R \in P$ with arbitrary $v_R \in R$
					\State $\Gamma \gets \Gamma + c$ \Comment{$c$ is now connected}
					\State $C' \gets C' - \{ c \mid P \text{ contains a $c$-colored region}\}$ \Comment{Remove eliminated colors}
				\EndFor
			\EndWhile
			\State $c \gets $ an arbitrary value from $C'$ \Comment{This will be the last remaining color}
			\State \Return $S$ + a move $(v_R, c)$ for each remaining non-$c$-colored color component
		\EndProcedure
	\end{algorithmic}
\end{algorithm}
We state the main result of this section in the following theorem.
\begin{restatable}{theorem}{ThmMinimalStrategy}
	\label{thm:minimal_strategy}
	For every graph $G$ with initial {\flooding} $f_0$, Algorithm~\ref{alg:minimal_strategy} computes in polynomial time a free strategy of $|C|-1$ moves for $G$, if one exists, and $\bot$, otherwise.
\end{restatable}
While the above result is only concerned with the free flood filling game, it is easy to see that a very similar but far simpler approach works for {\FIt}, too.
There, we would only have a single connected region from which we would be allowed to flood the graph.
So, we could only ever use one type of connecting move for a color $c$: flood the fixed pivot with $c$.
If, at any point, while greedily applying these moves, we do not have a single color available that can be connected by a move like that, the graph cannot be fixed-flooded in $|C| - 1$ moves.

In the context of free flood filling $(\pfive,\bull,\gem)$-free graphs, however, the above result settles the \emph{edge}-case, as formalized in the following corollary.

\begin{corollary}
	\label{cor:edge_case}
	For every graph $G$ that has a dominating edge and for any initial {\flooding} $f_0$, a shortest free strategy for $G$ can be computed in polynomial time.
\end{corollary}
Notice that our result is effectively stronger than required for our case study.
In fact, any graph class that has an upper bound on the length of free strategies of at most $|C|$ steps is in polynomial time by the same arguments.


\section{Monochromatic Bead Bracelets in Polynomial Time}
\label{sec:Polytime_Bead_Bracelets}

The final step of our approach is to compute optimal free strategies for $5$-bead bracelets in polynomial time.
Indeed, this component of our approach originates from the first author's master's thesis~\cite{masterthesis}.
From a conservative perspective, this problem is already resolved, as an elementary observation shows that $5$-bead bracelets are \cogem-free.
However, bead bracelets exhibit significantly richer structural properties than the mere presence of dominating {\pfour}.
This makes them a worthwhile object of study for more efficient algorithms -- the method from \cite{rosenke2026} works in polynomial time but with a prohibitively large exponent.
Additionally, we hope that our analysis of this case leads to further insights toward potential generalizations.

As outlined in Section~\ref{sec:graph_structure}, the method nevertheless builds on the {\cfive}-hub structure to quickly capture a small dominating set and, from there, flood the entire graph color by color.
Interestingly, this can be achieved in an even simpler way than in the \emph{edge} case.
In the remainder of this section, let $G$ be a $5$-bead bracelet constructed from beads $H_0, H_1, \dots, H_4$.
For readability, all indices used to refer to beads and related graph elements are taken modulo five without further mention.

The key observation is that disconnected colors cannot be located as arbitrarily as in the previous section.
If a color $c$ is, with respect to $f_0$, disconnected in $G$, then there exist vertices $v$ and $w$ in distinct color components of $f^{-1}_0(c)$.
Since vertices from neighboring beads are adjacent, it follows that $v$ and $w$ must lie either in the same bead or in non-neighboring ones.
The fundamental {\cfive} of $G$ leaves little room and, consequently, vertices $v$ and $w$ from non-adjacent beads always have distance exactly two.

A second central insight is that, for every $k \in \{0,1,\dots,4\}$, any vertex triple $u \in V(H_k)$, $v \in V(H_{k+1})$, and $w \in V(H_{k+2})$ taken from three consecutive beads forms a dominating set in $G$.
By the full join of neighboring beads, we have
$H_{k-1} + H_{k+1} \subseteq N(u)$,
$H_{k} + H_{k+2} \subseteq N(v)$, and
$H_{k} + H_{k+3} \subseteq N(w)$,
and hence $N[u] + N[v] + N[w] = V(G)$.
As in the previous section, this yields $|C|+1$ as an upper bound on the length of a shortest strategy for $G$.
Indeed, at most two moves suffice to conquer the triple $u, v, w$, and from this dominating set, $|C|-1$ further moves are sufficient to flood every color except one.

As above, the main objective of an algorithm is to distinguish between the possible length cases of a shortest strategy, namely $|C|-1$, $|C|$, or $|C|+1$ moves.
In fact, we show that the third case never occurs.
Before doing so, however, we need to decide whether a given $5$-bead bracelet $G$ with initial {\flooding} $f_0$ can be made monochromatic in the optimal $|C|-1$ moves.
While this could be determined using Algorithm~\ref{alg:minimal_strategy} from Section~\ref{sec:Polytime_Edge_Case}, the situation here is considerably simpler, again.
We therefore give a direct characterization.

\begin{restatable}{lemma}{OptimalBeadStrategy}
\label{lem:bead_CminusOne_strategy}
	If $G$ is a bracelet of $5$ beads $H_0, H_1, \dots, H_4$, then, on initial {\flooding} $f_0$, there exists a free strategy of $|C|-1$ moves for $G$ if and only if there are distinct $i, j \in \{0,1,\dots,4\}$ such that, with respect to $f_0$, the bead $H_i$ contains a connected color and
	\begin{itemize}
		\item $H_j$ contains a (not necessarily different) connected color, or
		\item $H_j$ contains a disconnected color $c$ such that $f^{-1}_0(c) + V(H_i)$ is connected in $G$.
	\end{itemize}
\end{restatable}

If $G$ and $f_0$ do not admit a free strategy of $|C|-1$ moves, then we show that there always exists a strategy of length $|C|$.
The underlying idea is the following.
Either a dominating three-vertex path can be conquered in a single move, after which $|C|-1$ additional moves suffice to flood all colors, or conquering such a dominating path requires two moves but, in doing so, one color can be eliminated simultaneously.
In the latter case, only $|C|-2$ further moves are needed to flood the remainder of the graph.

More precisely, our polynomial-time approach for free-flooding $5$-bead bracelets works as follows.
It begins by checking the conditions for the existence of a strategy with $|C|-1$ moves constructing such a strategy if they are satisfied.
If this attempt fails, the algorithm checks if any two non-neighboring beads -- say $H_0$ and $H_2$ -- contain equally colored vertices $u \in V(H_0)$ and $w \in V(H_2)$.
Then a dominating path $D = u,v,w$ can be created in a single move by flooding the color $f_0(u)$ from any pivot vertex $v \in V(H_1)$.

Otherwise, all non-neighboring beads are color-disjoint, and we require two moves to produce $D = u,v,w$:
first, flooding the color $f_0(u)$ for some $u \in V(H_0)$ from an arbitrary pivot vertex $v \in V(H_1)$ and, second, flooding the color $f_0(w)$ for some $w \in V(H_2)$ from the same pivot $v$.
With respect to $f_1$, the {\floodarea} $f^{-1}_1(c)$ forms a single color component containing $v$, so that the second move recolors all vertices of this component.
Thus, these two moves already eliminate the color $c = f_0(u)$.

Using the dominating path $D$ as a hub, the remaining colors can then be flooded in the standard way.
Consequently, the algorithm correctly returns a free strategy of $|C|$ moves.

\begin{restatable}{lemma}{BeadAlgorithm}
\label{lem:bead_polytime}
	For every $5$-bead bracelet $G$ with initial {\flooding} $f_0$, a shortest free strategy for $G$ can be computed in polynomial time.
\end{restatable}
Together with Corollary~\ref{cor:edge_case}, the above lemma implies our main result stated in Theorem~\ref{thm:polynomial_time+cases}.

We would like to point out that the results of this section are, in fact, stronger than what is expressed by the concluding lemma.
Indeed, it is not necessary for the graph $G$ to be a bracelet whose beads $H_0,H_1,\dots,H_4$ are cographs.
Since the cograph structure of the beads is never used, our approach applies to any choice of non-empty graphs $H_0,H_1,\dots,H_4$, provided they are arranged as a bracelet over a {\cfive}.


\section{Flood Filling Games are $\NP$-complete on Thin-Spider {\Playboards}}
\label{sec:np_completeness}

This section proves Theorem~\ref{thm:NP-completeness:from_split_graphs}.
Our approach adapts ideas from the methods of Fleischer and Woeginger~\cite{Fleischer2010}, Fukui~et~al.~\cite{Fukui2013}, and Fellows~et~al.~\cite{FELLOWS2015}.
However, the reduction gadgets used in our proof are so-called \emph{thin spiders}.
Thin spiders, as defined by Nastos~and~Gao~\cite{thinspiders}, are connected split graphs admitting a vertex partition into a clique $K$ and an independent set $I$ such that every vertex in $K$ and in $I$ has at most one neighbor in the respective other part.

It is easy to see that these graphs are exactly the {\claw}- and {\dimond}-free split graphs.
Indeed, any induced {\claw} would have its central vertex in $K$, forcing it to be adjacent to at least two vertices in $I$.
Similarly, any induced {\dimond} would place at least one of the two non-adjacent vertices into $I$, giving it at least two neighbors in $K$.
The converse direction is obvious, if $|K| \leq 2$.
Otherwise, let $G$ be a split graph in which some vertex $v$ has at least two neighbors $x,y$ in the opposite part.
Then we find an induced {\claw} or {\dimond} in $G$ as follows.
We may assume that the partition into $K$ and $I$ is chosen such that every vertex of $I$ is non-adjacent to at least one vertex of $K$, as otherwise it could simply be moved from $I$ into $K$.
Firstly, if $v\in I$ and $x,y\in K$, then there exists a vertex $z\in K-N[v]$, and the vertices $v,x,y,z$ induce a {\dimond}.
And, secondly, if $v\in K$ and $x,y\in I$, let $z\in K-N[x]$.
Since $v$ is adjacent to $z$, while $y$ and $z$ are not adjacent (as this would imply $|N[y]\cap K|>1$), the vertices $v,x,y,z$ induce a {\claw} with central vertex $v$.

Notice that thin spiders exhibit very little structure: up to isomorphism, there is exactly one thin spider for each value of $|K|$ and $|I|$. 
In this sense, our result is slightly stronger than the earlier $\NP$-hardness results for split graphs from~\cite{Fleischer2010,Fukui2013}.
On the other hand, these works establish hardness even for instances with a proper initial {\flooding} (where no pair of adjacent vertices has the same color), which strengthens their conclusions.
In our setting, this restriction has to be dropped, as the simplicity of the thin-spider gadgets does not provide sufficient structure to compensate for it.

With these deeper insights into thin spiders, we are now ready to state the main result of this section.
\begin{restatable}{lemma}{ThinSpiderNPComplete}
\label{lem:NP-completeness:thin_spiders}
    {\MinimumFloodIt} and {\MinimumFreeFloodIt} are $\NP$-complete on thin-spider graphs.
\end{restatable}
The proof of the lemma essentially works by a standard reduction from \textsc{Minimum Vertex Cover} to both flooding games.
The basic idea follows the $\NP$-hardness proof for trees by Fellows~et~al.~\cite{FELLOWS2015}.
The vertices of a given graph $H$ are used as colors in the initial {\flooding} of a gadget graph~$G$; in our case, a thin spider, and in~\cite{FELLOWS2015}, a tree of diameter four.
Both gadget types share the principle that every edge $vw$ of $H$ is represented by two leaves, one corresponding to each endpoint~$v$ and~$w$.

Then $H$ has a vertex cover $W$, if and only if all vertices of~$G$ can be conquered from a designated, fixed pivot $p$ using exactly $|V(H)| + |W|$ moves.
These moves are divided into three stages: the first runs through the colors of~$W$, the second through $V(H) - W$, and the third through~$W$ again.
The first stage provides a head start for every vertex $w \in W$, allowing it to conquer the parent vertex $w$ for all edges $vw$ of~$H$.
This rapid expansion is exploited in the second stage: while conquering the parents of all remaining leaves, every leaf representing the $v$-endpoint of an edge $vw \in E(H)$ with $w \in W$ is flooded as well.
At this point, the graph $G$ is monochromatic except for the leaves representing the $w$-endpoints of edges of~$H$ with $w \in W$, which are conquered in the third stage.

An additional difficulty arises for {\FreeFloodIt}, where every \emph{free} pivot has to be effectively \emph{forced} to behave like the fixed pivot $p$.
To this end, we adapt an idea of Fukui~et~al.~\cite{Fukui2013} and attach additional leaves to~$G$, one for each color in~$V(H)$, whose parents lie in a designated part of~$K$ associated with the intended fixed pivot $p$.
This construction restricts the available flooding moves and thereby allows them to be interpreted as if they were performed from~$p$.

We conclude this section by deriving the announced consequences of Lemma~\ref{lem:NP-completeness:thin_spiders}.
Since thin spiders are exactly the $(\claw,\dimond)$-free split graphs, and since split graphs are the $(\twoKtwo,\cfour,\cfive)$-free graphs~\cite{foldes}, we obtain the following corollary.
\begin{corollary}\label{cor:NP-completeness:thin_spiders}
    {\MinimumFloodIt} and {\MinimumFreeFloodIt} are $\NP$-complete on $(\claw,\dimond,\twoKtwo,\cfour,\cfive)$-free graphs.
\end{corollary}
We observe that {\pfive} and {\cobanner} contain an induced {\twoKtwo}, the {\gemlets} {\house} and {\banner} contain an induced {\cfour}, {\chair} contains an induced {\claw}, and both {\kite} and {\gem} contain an induced {\dimond}.
Hence, the corollary above implies Theorem~\ref{thm:NP-completeness:from_split_graphs}.


\section{Conclusion}
\label{sec:conclusion}

To date, the complexity of {\FloodIt} and {\FreeFloodIt} has been determined for $896$ of the $1024$ cograph generalizations.
As a central contribution, this paper added $192$ classes to this state of the art.
The $128$ unresolved classes consist of (i) $(\pfive, \bull)$-free graphs containing a {\gem} or (ii) $(\bull,\chair)$-free graphs.
Based on a preliminary analysis of structural restrictions imposed by these forbidden subgraphs, and the apparent absence of typical hardness conditions within these classes, we arrive at the following conjecture.
\begin{conjecture}
	{\MinimumFloodIt} and {\MinimumFreeFloodIt} are solvable in polynomial time on all $128$ remaining classes.
\end{conjecture}

We note that our results are, in fact, slightly stronger than suggested by the classification in terms of forbidden \gemlets.
While $(\pfive,\bull,\gem)$-free graphs exhibit the structure derived in Section~\ref{sec:graph_structure}, the converse does not hold.
Since our algorithms effectively operate on either a dominating edge or a {\cfive}-bracelet on arbitrary graphs, the resulting tractability extends beyond $(\pfive,\bull,\gem)$-free graphs.
This observation suggests that flooding tractability is governed by the existence and\,/\,or fine structural properties of dominating subgraphs that serve as flooding hubs.
Identifying and exploiting such structures may therefore provide promising directions for extending tractability to larger graph classes and for achieving a deeper understanding of flooding games in general.


\section{Technical Proofs}
\label{sec:proofs}

This section collects the proofs of all the theorems and lemma from the main text, presented in the same order as they appear. Each result is restated for clarity before its proof.

\CFiveCographInjection*
\begin{proof}
    Consider any induced {\cfive} in $G$ and let $v_0$, $v_1$, $v_2$, $v_3$, $v_4$ be the corresponding vertices.
    We define $N_i = N(v_{i-1}) \cap N(v_{i+1})$ for all $i \in \{0, \dots, 4\}$ where, like in the remainder of the proof, indices outside $\{0, 1, 2, 3, 4\}$ are considered modulo five.
    To prove the lemma, we are required to show that $H_0 = G[N_0], \dots, H_4 = G[N_4]$ are cographs whose vertex sets $N_0, \dots, N_4$ form a partition of $V(G)$ and that edges going between them follow the {\cfive}, that is, $G$ has an edge $xy$ between $x \in N_i$ and $y \in N_j$ if and only if $0 \leq i < j < 5$ with $j-i \in \{1,4\}$.

    \medskip
    We begin with showing that the $N$-sets form a partition and, to that end, prove first that $N_0, \dots, N_4$ are mutually disjoint.
    Any setting where a vertex $x$ is in $N_i \cap N_j$ with $0 \leq i < j < 5$ would mean by definition that $x \in N(v_{i-1}) \cap N(v_{i+1})$ and $x \in N(v_{j-1}) \cap N(v_{j+1})$.
    If $j-i \equiv 1 \mod 2$ then $x$ would be the {\tip} of an induced {\gem} with {\pfour} on $v_{i-1},v_i,v_{i+1},v_{i+2}$ or on $v_{i-2},v_{i-1},v_i,v_{i+1}$.
    Otherwise, if $j-i \equiv 0 \mod 2$, there is an induced {\bull} with {\tip} $x$ and {\pfour} on $v_i,v_{i-1},v_{i-2},v_{i+2}$ or on $v_i,v_{i+2},v_{i+2},v_{i-2}$.
    
    Before we go on in proving that the $N$-sets partition $V(G)$, we insert an observation that best fits at this point.
    In fact, for every $i \in \{0,\dots,4\}$, we can say for all $x \in N_i$ that, while being adjacent to $v_{i-1}$ and $v_{i+1}$ by definition, $x$ is not adjacent to $v_{i-2}$ or $v_{i+2}$.
    That happens, because, if $x$ was adjacent to $v_{i-2}$ then $x \in N(v_{i-2}) \cap N(v_{i+1}) = N_{i+2}$ and, thus, $x \in N_i \cap N_{i+2}$, which is a contradiction.
    Similarly, $x \in N(v_{i+2}) \cap N(v_{i-1}) = N_{i-2}$ leads to the contradiction $x \in N_i \cap N_{i-2}$.
    
    Now, because the $N$-sets are already known to be mutually disjoint, it is sufficient to prove that all vertices are covered by them to establish them as a partition of $V(G)$.
    For a contradiction, we, hence, assume a vertex $v \in V(G) - (N_0 + N_1 + N_2 + N_3 + N_4)$ exists.
    By the connectedness of $G$, there exists a path from $v$ to $w$, where $w$ is the only vertex on the path that is contained in $N_0 + N_1 + N_2 + N_3 + N_4$.
    Let $x$ be the unique neighbor of $w$ on this path (possibly with $x = v$).
    From the selection of $x$, we know that $x \not\in \{v_0,\dots,v_4\}$.
    Clearly, $x$ cannot be adjacent to two or more vertices of $v_0,\dots,v_4$.
    In fact, if $x$ was (among others) adjacent to $v_i$ and $v_j$ with $0 \leq i < j < 5$ and $j-i \in \{2,3\}$ then $x \in N_{i+1}$, which is a contradiction to the choice of $x$.
    Otherwise, if $x$ saw exactly $v_i$ and $v_j$ on the {\cfive} with $i \in \{0, \dots, 4\}$ and $j = i+1$ then there would be an induced {\bull} with {\pfour} on $v_{i-1}, v_i, v_{i+1}, v_{i+2}$ and {\tip} $x$.
    Next, we see that $x$ cannot be adjacent to exactly one vertex $v_i$ of the {\cfive}, either.
    This would induce a {\pfive} on $x,v_i,v_{i+1},v_{i+2},v_{i-2}$.
    In the last case, $x$ does not see any vertex $v_0,\dots,v_4$.
    However, $x$ is adjacent to $w$, which is in $N_i$ for some $i \in \{0, \dots, 4\}$.
    Recall that this means $v_{i-1}w, v_{i+1}w \in E(G)$, while the edges $v_{i-2}w$ and $v_{i+2}w$ are not present.
    This would mean that there is an induced {\pfive} on $x,w,v_{i+1},v_{i+2},v_{i-2}$, a contradiction.
    Hence, $v$ does not exist and the $N$-sets partition $V(G)$.

    \medskip
    The next step is to show for all $x \in N_i$ and $y \in N_j$ with $0 \leq i < j < 5$ that $xy \in E(G)$ if and only if with $j-i \in \{1,4\}$.
    So, firstly, let $j-i \in \{1,4\}$, hence, with $v_i$ and $v_j$ neighbors on the {\cfive}.
    We show that $xy \in E(G)$.
    If $x = v_i$ then, from $y \in N_j = N(v_i) \cap N(v_{i+2})$, we get $y \in N(x)$.
    Symmetrically, $y = v_j$ implies $x \in N(y)$.
    So, consider $x \in N_i - v_i$ and $y \in N_j - v_j$ and assume that $xy \not\in E(G)$.
    By definition and the above observation, we know that $x$ sees $v_{i-1}$ and $v_j$ but neither $v_{i-1}$ nor $v_{j+1}$, while $y \in N(v_i) \cap N(v_{j+1})$ cannot see $v_{i-1}$ or $v_{j+2}$.
    That $xv_i, yv_j \not\in E(G)$ is implied by the otherwise induced {\bull} with {\pfour} on $v_{i-2}, v_{i-1}, v_i, y$ and {\tip} $x$ or with {\pfour} on $x, v_j, v_{j+1}, v_{j+2}$ and {\tip} $y$.
    But this already contradicts the assumption because, now, there is an induced {\pfive} on $v_{i-1},x,v_j,v_{j+1},y$.

    Secondly, we assume, for a contradiction, that $xy \in E(G)$ where $x \in N_i$ and $y \in N_j$ with $j-i \in \{2,3\}$.
    By symmetry, we can say, without loss of generality, that $j = i+2$.
    If $x = v_i$ then, since $y = v_j$ is impossible here, $y \in N(v_i) \cap N(v_{j-1}) \cap N(v_{j+1})$.
    Thus, $y$ would be in two $N$-sets, $N_j$ and $N_{i-1}$, which has been excluded, above.
    In the same way, we find $y \not= v_j$.
    So again, consider $x \in N_i - v_i$ and $y \in N_j - v_j$ and, hence, $x \in N(v_{i-1}) \cap N(v_{i+1})$ cannot see $v_{i-2}$ or $v_{i+2}$, while $y \in N(v_{j-1}) \cap N(v_{j+1})$ cannot see $v_{j-2}$ or $v_{j+2}$.
    If the edge $xv_i$ existed in $G$ then there would be an induced {\gem} with {\pfour} on $v_{i-1}, v_i, v_{i+2}, y$ and {\tip} $x$.
    In the absence of $xv_i$, however, there is an induced {\bull} with {\pfour} on $v_i,v_{i+1},y,v_{i-2}$ and {\tip} $x$.
    So this assumption fails, too, and we have no edges between $N_i$ and $N_j$.

    \medskip
    The last step, establishes $H_0, \dots, H_4$ as cographs, that is, {\pfour}-free graphs.
    This is obvious, again, since any {\pfour} in $H_i$ would induce a {\gem} with {\tip} $v_{i+1}$.
    This completes the proof.
\end{proof}

\EdgeCase*
\begin{proof}
	If $G$ is a cograph then, as it is connected, there are non-empty cographs $G_1$ and $G_2$ such that $G$ is the full join of $G_1$ with $G_2$.
	Every edge $vw$ of $G$ with $v \in V(G_1)$ and $w \in V(G_2)$ fulfills $N[v] \supseteq V(G_1)$ and $N[w] \supseteq V(G_2)$ and, so, $N[v] + N[w] = V(G)$.

	\medskip
    Otherwise, $G$ contains at least one {\pfour}.
	Before we can prove the lemma in this case, we argue that $G$ contains a dominating {\pfour}, that is, one that has every vertex of $G$ as a neighbor to at least one of its own vertices.
	To see this, let us fix any {\pfour}, say $\phi$, in $G$, say on vertices $a,b,c,d$.
    Then, we only have to consider the case where the vertex set $X = V(G) - N[\phi]$ is not empty.
    
    We first show that every $x \in X$ is already close to $\phi$, that is, $\min\{\dist_G(x,v) \mid v \in V(\phi)\} = 2$, where $\dist_G(x, v)$ is the length of a shortest path between vertices $x$ and $v$ in $G$.
	By definition of $X$ we already have $\min\{\dist_G(x,v) \mid v \in V(\phi)\} \geq 2$, since $x \not\in N[\phi]$.
	To see that this is also an upper bound let us fix any shortest path between $x$ and a vertex $u \in V(\phi)$, which exists due to the connectedness of $G$.
	Note that this path contains exactly one vertex from $\phi$ as otherwise there is a choice for $u$ that results in a shorter path than the one we fixed.
    Let $y$ be the vertex on this path that is second to last, that is, the predecessor of $u$.
    By $x \not\in N[\phi]$ we have $y \not= x$ and since $u$ is the only vertex on this path from $\phi$ we also have $y \not\in V(\phi)$.
    Furthermore, $y$ is not adjacent to all of $a,b,c,d$ as this would induce a {\gem} in $G$.
    So, there is an edge $vw \in E(\phi)$ with $v$ adjacent to $y$ (possibly $u=v$) and $w$ not adjacent to $y$.
    If the distance between $\dist_G(x,y) \geq 2$ then there is an induced path between $x$ and $w$ of length at least five.
    This is forbidden and, hence, $x \in N(y)$.

    Next, we show for every vertex $x \in X$ that all $y \in N(x) \cap N(\phi)$ are either adjacent to $a$ and $c$ but none of $b$ and $d$, or the other way around, that is, to $b$ and $d$ but not to $a$ or $c$.
    That $y$ is adjacent to all four vertices means that there is an induced {\gem}, impossible.
    But, if only one vertex of $\phi$ was adjacent to $y$, without loss of generality either $a$ or $b$, then this would induce a {\pfive}, either on $x,y,a,b,c$ or on $x,y,b,c,d$.
    With exactly three neighbors of $y$ in $V(\phi)$, again without loss of generality, say $a,b,c$ or $a,b,d$, we get an induced {\bull}, either with {\pfour} on $x,y,c,d$ and {\tip} $b$ or on $x,y,b,c$ and {\tip} $a$.
    This leaves us with exactly two $\phi$-vertices in $N(y)$, either $a,b$ or $a,c$ or $a,d$ or $b,c$ or $b,d$ or $c,d$, where the second and the second to last pair are the allowed ones.
    The pairs $a,b$ and $a,d$ and $c,d$ are impossible because they induce a {\pfive}, either on $x,y,b,c,d$ or on $x,y,a,b,c$ or on $x,y,c,b,a$.
    Finally, $b,c$ induces a {\bull} with $\phi$ as its {\pfour} and {\tip} $y$.
    Hence, either $N(y) \cap V(\phi) = \{a,c\}$ or $\{b,d\}$.

    The previous arguments allow to map every vertex $x \in X$ to a vertex $y(x) \in N(x) \cap N(\phi)$ that is either adjacent to $a,c$ but not $b,d$, or to $b,d$ but not $a,c$.
    Here, we call $y(x)$ the \emph{guard} of $x$.
    Subsequently, we show that this mapping can be rather constant.
    In fact, if there was a vertex $y \in N(a) \cap N(c)$ or $y \in N(b) \cap N(d)$ such that $X \subseteq N(y)$ then we can define the unique guard $y(x) = y$ for all $x \in X$.
    Otherwise, we provide a single edge $y_1y_2 \in E(G)$ with $y_1 \in N(a) \cap N(c)$ and $y_2 \in N(b) \cap N(d)$ such that $X \subseteq N(y_1) + N(y_2)$.
    For this purpose, let $X_1 = \{x \in X \mid y(x) \in N(a) \cap N(c)\}$ consist of all vertices with guard adjacent to $a$ and $c$ and, analogously, let $X_2 = \{x \in X \mid y(x) \in N(b) \cap N(d)\}$.
    Moreover, for every $x \in X$, we let $F_1(x) = N(y(x)) \cap X_1$ and $F_2(x) = N(y(x)) \cap X_2$ be the \emph{fellows} of $x$ within $X_1$ and $X_2$, respectively, that is, the vertices that allow for the same guard as $x$.
    
	Then, we argue that the vertices within $X_1$ and, respectively, $X_2$ can be ordered according to fellowships.
    In fact, for all $x, x' \in X_1$ it is true that $F_1(x) \subseteq F_1(x')$ or $F_1(x) \supseteq F_1(x')$ and, analogously, the same works for $X_2$.
    More precisely, if there existed $z \in F_1(x) - F_1(x')$ and $z' \in F_1(x') - F_1(x)$ then, with $y = y(x)$ and $y' = y(x')$, they would induce a {\pfive}, either on $z,y,a,y',z'$, if $zz', yy' \not\in E(G)$, or on $z',z,y,c,d$, if $zz' \in E(G)$, or they would induce a {\bull} with {\pfour} on $z,y,y',z'$ and {\tip} $a$, if $zz' \not\in E(G)$ and $yy' \in E(G)$.
    A similar argument works for $X_2$ and $F_2$.
    Consequently, there have to be maximal elements $z_1 \in X_1$ and $z_2 \in X_2$ such that $F_1(z_1) \supseteq F_1(x)$ for all $x \in X_1$ and $F_2(z_2) \supseteq F_1(x)$ for all $x \in X_2$, respectively.
    Clearly, $F_1(z_1) = X_1$ as, otherwise, there would be a vertex $x \in X_1$ with a guard $y(x)$ that is not adjacent to $z_1$, thus, with $F_1(z_1) \not\supseteq F_1(x)$.
    In the same way, $F_2(z_2) = X_2$.
    Hence, the guard vertices $y_1 = y(z_1)$ and $y_2 = y(z_2)$ can operate as guard for the entire sets $X_1$ and $X_2$, respectively.
    Before we continue, we select $x_1 \in X_1 - N(y_2)$ and $x_2 \in X_2 - N(y_1)$.
    Such vertices exist for the following reason.
	If all $x \in X_1$ are adjacent to $y_2$ then $X \subseteq N(y_2)$.
	Similarly, $X_2 \subseteq N(y_1)$ implies $X \subseteq N(y_1)$.
	Either way, we are in the case where a single vertex $y$ exists with $X \subseteq N(y)$.
    With these vertices $x_1$ and $x_2$ it is easy to see that $y_1y_2 \in E(G)$, as otherwise, there would be a {\pfive} on $x_1,y_1,c,b,y_2$.

    We are finally ready to select a new {\pfour}, we will call it $P$, with $V(G) = N[P]$.
    Firstly, if there is a single guard $y$ with $X \subseteq N(y)$, without loss of generality, say $a, c \in N(y)$ but neither $b$ nor $d$, then we choose $P$ on the vertices $a, y, c, d$.
    Clearly, $P$ is an induced {\pfour} in $G$.
    We show that every vertex $z \in V(G)-V(P)$ is in $N(P)$.
    Since $X \subseteq N(y)$ and $ab \in E(G)$, we do not have to consider $z \in X$ or $z = b$.
    By construction, we, thus, know that $z \in N(\phi) - V(\phi)$, a vertex outside but adjacent to $\phi$.
    If any edge of $za$, $zc$, $zd$ and $zy$ is present then $z$ is $N(P)$.
    In the final case, where $z$ is just adjacent to $b$, any edge $xz$ with $x \in X$ would imply $zd \in E(G)$ and, thus, $z \in N(P)$.
    But $X \cap N(z) = \emptyset$ is impossible because there is an induced {\pfive} on $z,b,c,y,x$.

    Secondly, if there is an edge $y_1y_2$ of two guards, $y_1 \in N(a) \cap N(c)$ and $y_2 \in N(b) \cap N(d)$, then we select $P$ on the vertices $a,y_1,y_2,d$, which obviously induces a {\pfour}.
    Again, to show $N[P] = V(G)$, consider $z \in V(G) - V(P)$.
    If $z \in X+\{b,c\}$ then $z \in N(y_1) + N(y_2)$ and if one of the edges $za$, $zd$, $zy_1$, $zy_2$ exists then $z \in N(P)$.
    Otherwise, $z$ would be a vertex outside $\phi$ but adjacent to $b$ or $c$ or both.
    In any case, $z$ could not be adjacent to some $x \in X$ because, like before, this would put $z$ into $X_1$ or $X_2$.
    If $z$ was adjacent to $b$ and\,/\,or $c$ then this would induce a {\pfive}, either on $z,b,a,y_1,x_1$ or on $z,c,d,y_2,x_2$ with $x_1 \in X_1-N(y_2)$ and $x_2 \in X_2-N(y_1)$, as selected above.

	\medskip
	Now that we have with $P$ a dominating {\pfour}, let the vertices of $P$ be denoted as $u,v,w,x$.
	We show that the middle edge $vw$ of $P$ alone fulfills $V(G) = N[v] + N[w]$, already.
	In fact, assume a vertex $y \in V(G) - N[v] - N[w]$.
	Since $P$ is dominating, at least one of the edges $uy$ and $xy$ must be present.
	However, if one of them was missing then there would be an induced {\pfive}, either on $y,u,v,w,x$ or on $u,v,w,x,y$.
	But both edges present mean an induced {\cfive} on $u,v,w,x,y$, which is also impossible.
	Hence, $y$ cannot exist, which completes the proof.
\end{proof}

\ThmMinimalStrategy*
\begin{proof}
	We first show that Algorithm~\ref{alg:minimal_strategy} always terminates.
	If all colors are connected in $f_0$, this is obviously the case.
	The algorithm also clearly terminates if $\Delta = \emptyset$ in any iteration of the while-loop.
	For all other cases, notice that $\Gamma$ is always extended by all elements in $\Delta \subseteq C$ and $\Delta \cap \Gamma = \emptyset$.
	In other words, each iteration of the while-loop (that does not return $\bot$) assigns to $\Gamma$ a larger subset of $C$.
	Since $C$ is finite, we must have $\Gamma = C$ after a finite number of iterations, exiting the while-loop.
	Thus, the algorithm terminates for all inputs.

	\medskip
	We now show that Algorithm~\ref{alg:minimal_strategy} returns a strategy if a minimal free strategy exists, so assume there is a free strategy of $|C| - 1$ moves.
	For the sake of contradiction, assume also that Algorithm~\ref{alg:minimal_strategy} returns $\bot$ when given $G$ and $f_0$ as input.
	Then $f_0$ contains at least one color that is not connected, since otherwise the algorithm would skip the while-loop and return a strategy.
	Furthermore, there was some iteration of the while-loop where the set $\{ c \in C \mid \Gamma \text{ connects } c\}$ was empty.
	Let $I$ be the free move order computed by our algorithm up to this point, let $S$ be an optimal free strategy, and let $c$ be the first color that $S$ connects, but $I$ does not.

	Note that $c$ exists, because $S$ must connect all colors.
	Were that not the case for some color $c'$, we could neither eliminate $c'$ nor could the strategy cause the entire graph to become $c'$-colored.
	The only way to connect two (or more) disconnected $c$-colored regions while also eliminating a color during each move is by flooding connected colors with $c$.
	Flooding disconnected colors does not eliminate a color, and flooding with any color other than $c$ can not create monochromatic $c$-colored paths between previously disconnected $c$-colored vertices.
	Thus, $S$ must contain moves $(v, c)$, where $v$ is the vertex of the single color component of a connected color.
	Let $R$ be the set of these color components.
	We define $W := \bigcup_{A \in R} \{f_0(v) \mid v \in A\}$.
	Clearly, $c \in \{c' \in C \mid W \text{ connects } c'\}$ and $S$ connects all colors in $W$ before $c$.

	Since $c$ is the first color that was not connected by $I$ it follows that $I$ connects all colors in $W$.
	Whenever a color is connected, it is added to $\Gamma$ (see line 13).
	But then $c \in \Delta$, since $W \subseteq \Gamma$ connects $c$.
	A contradiction.

	\medskip
	It remains to show that the returned move order is also a free strategy of length $|C| - 1$.
	To that end, let $S$ now be the move order returned by Algorithm~\ref{alg:minimal_strategy}.
	It follows from our arguments above, that $S$ must connect all colors eventually.
	Lines~17~and~18 in Algorithm~\ref{alg:minimal_strategy} guarantee that all remaining color components are flooded with the same color.
	The result of this is obviously a monochromatic flooding.
	In other words, $S$ is a strategy.

	To see that $S$ also only contains $|C| - 1$ moves, observe that any move added to $S$ is specifically crafted to only flood the single color component of a connected color.
	As long as no such move floods a $c$-colored color component with color $c$ by accident, we are done.
	But this is the case:
	$P$ contains no region of color $c$ by definition in any iteration $c$ for the loop in line~10, and line~18 specifically leaves out any $c$-colored color component when finalizing the graph with color $c$.
	As this implies that every move eliminates a color, we have exactly $|C| - 1$ moves in $S$.

	\medskip
	Finally, we observe that Algorithm~\ref{alg:minimal_strategy} has polynomial runtime.
	Using any standard disjoint-set data structure, we can keep track of each monochrome region by first representing each vertex through a singleton set and, unionizing at the beginning all sets of equal color.
	From this, we can extract the set of connected colors in polynomial time.
	By our arguments for termination of the algorithm, we can see that the while-loop has at most $|C|$ iterations, each of which is in polynomial time by using union-find and BFS.
\end{proof}

\OptimalBeadStrategy*
\begin{proof}
	Assume first that $G$ has a free strategy $S = (x_1, c_1), \dots, (x_s, c_s)$ of $s = |C|-1$ moves.
	Like in the previous section, every move has to eliminate one color so that, at the end, $G$ becomes monochromatic.
	So, with respect to $f_0$, at least $c_1$ has to be a connected color.
	We let $i$ be the index of the bead $H_i$ that contains $x_1$.

	If there is any connected color $c$ (possibly with $c=c_1$) where $f^{-1}_0(c) \cap V(H_j)$ is not empty for any $j \not= i$ then the statement of the lemma holds, already.
	So, we are left with the case where all connected colors are confined within $H_i$ and all vertices of other beads are part of disconnected colors.

	So, assume all colors $c$ leave $f^{-1}_0(c) + V(H_i)$ disconnected.
	Then consider any vertex $v \in V(H_a)$ with $a = i+1$ and let $c = f_0(v)$ be the corresponding disconnected color.
	By our assumption, $f^{-1}_0(c) + V(H_i)$ is not connected and, since neighboring beads are fully joined, it must be true for all $k \not= i$ that $f^{-1}_0(c) \cap V(H_k)$ is non-empty if and only if $k=a$ or $k=a+2$.
	With the same arguments, we obtain for any vertex $w \in V(H_{b})$ with $b=i-1$ with initial color $c' = f_0(w)$ that, for all $k \not=i$, $f^{-1}_0(c') \cap V(H_k)$ is non-empty if and only if $k=b$ or $k=b-2$.
	Notice that one of the two conditions is true for all disconnected colors that do not appear in $H_i$ only.
	
	Recall, that every move $m$ must eliminate a color and that this is only possible with a pivot that is part of the single color component $f^{-1}_{m-1}(z)$ of a connected color $z$.
	Moreover, any disconnected color $c$ can only become connected via a move $(x_m, c_m)$ with $c_m = c$ and pivot $x_m$ making $f^{-1}_{m-1}(x_m) + f^{-1}_{m-1}(c)$ connected.
	As $f_0$ defines connected colors only within $H_i$, there is a prefix $(x_1, c_1), \dots, (x_p, c_p)$ that uses pivots $x_1, \dots, x_p$ from $H_i$ and where move $p$ is the first that conquers a vertex outside $V(H_i)$.
	The move $p$ floods with a color $c_p=f^{-1}_{p-1}(v)$ of a vertex $v \not\in V(H_i)$ that is adjacent to $f^{-1}_{p-1}(x_p)$, thus, with $v \in V(H_{i-1}) + V(H_{i+1})$.
	Without loss of generality, assume that $v$ belongs to $H_{i+1}$.
	Using the above arguments, we find at least one vertex $w \in f^{-1}_{p-1}(c_p) \cap V(H_{i+3})$.
	Due to the {\cfive}-form of $G$ and by the definition of $p$, the vertex $w$ is not adjacent to $f^{-1}_{p-1}(x_p)$ and, thus, $f^{-1}_{p}(c_p)$ remains disconnected.
	But the move $p$ has already used the color $c_p$ and, hence, there is no way for the remaining moves of $S$ to connect this color and still flood with every color but one.

	\medskip
	Now, for the converse direction, assume that $H_i$ and $H_j$ fulfill the condition of the lemma.
 	We construct a free strategy for $G$ of $|C|-1$ moves that, that in the first up to two moves conquers a dominating three-vertex-path $D$ while still eliminating one color per move, and that uses $D$ in the remaining moves to conquer the rest of the colors except one.
	So, in the following we just demonstrate how $D = u,v,w$ can be constructed in every case, and we argue each time that all corresponding moves eliminate a color.

	Let $c$ be a connected color such that $f^{-1}_0(c) \cap V(H_i)$ is not empty.
	Furthermore, let the connected or disconnected color in $H_j$ be $c'$ (possibly with $c=c'$).
	Then, in the first case, $H_i$ and $H_j$ are adjacent, say with $i + 1 = j$.
	We select arbitrary vertices $u \in V(H_i) \cap f^{-1}_0(c)$, $v \in V(H_{i+1}) \cap f^{-1}_0(c')$, and $w \in V(H_{i+2})$.
	If $c = c'$ (so $c'$ is also connected) then $D = u,v,w$ is conquered with the single move $(u, f_0(w))$, which eliminates $c$.
	Otherwise, the two moves $(u, c'), (u, f_0(w))$ conquer $D$ and eliminate $c$ and $c'$.
	This happens even if $c'$ is a disconnected color, since $f^{-1}_0(c') + V(H_i)$ is connected, which, in bead-bracelets, implies the same for $f^{-1}_0(c') + f^{-1}_0(u)$.

	In the second case, $H_i$ and $H_j$ are non-neighboring, without loss of generality, say with $i + 2 = j$.
	Again, we select $u \in V(H_i) \cap f^{-1}_0(c)$ and $w \in V(H_{i+2}) \cap f^{-1}_0(c')$ and see that $c = c'$ this time means that there is a $v \in V(H_{i+1}) \cap f^{-1}_0(c)$ that completes a monochromatic $D$ without any moves.
	Otherwise, and if $c'$ is connected, we select $v$ arbitrarily in $V(H_{i+1})$ and make the two moves $(u, f_0(v)), (w, f_0(v))$ to conquer $D$, which eliminate $c$ and $c'$.
	Finally, having $c'$ disconnected but under the condition stated in the lemma, requires the existence of a vertex $v \in V(H_{i+1}) \cap f^{-1}_0(c)$, since $G$ is a $5$-bead bracelet and $f^{-1}_0(c') + V(H_i)$ must be connected.
	Then the move $(u, c')$ establishes $D$ and eliminates $c$.
	This completes the proof.
\end{proof}

\BeadAlgorithm*
\begin{proof}
	An algorithm can firstly determine if a $|C|-1$-move strategy exists and, if so, compute one.
	If that fails, the algorithm efficiently produces a dominating monochromatic set $D$ and, based on that, produce a $|C|$-move strategy.
	
	For the first case, we check the conditions of Lemma~\ref{lem:bead_CminusOne_strategy}.
	In case a positive match is found, the proof of Lemma~\ref{lem:bead_CminusOne_strategy} describes the computation of the $|C|-1$-move strategy via a dominating set $D$.

	If no match is found, we check if there exist equally colored vertices $u \in V(H_k)$ and $w \in V(H_{k+2})$ in non-neighboring beads $H_k$ and $H_{k+2}$ for some $k \in \{0,1,\dots,4\}$.
	For this we consider all non-neighboring pairs of beads $H_k$ and $H_{k+2}$, which count five and testing all pairs of vertices $u \in V(H_k)$ and $w \in V(H_{k+2})$ for equal color.
	In the positive case, a dominating set $D$ can be made monochromatic in one move, and a corresponding strategy is produced in linear time.

	Otherwise, the algorithm selects in $\cO(1)$-time an arbitrary pivot $v \in V(H_1)$ and the colors $c = f_0(u)$ of any vertex $u \in V(H_0)$ and $c' = f_0(w)$ of any vertex $w \in V(H_2)$.
	The moves $(v,c)$ and $(v,c')$ make $u,v,w$ a monochromatic dominating set.
	But they also eliminate color $c$.
	This is because $c$ can only appear in $H_0$ and at most one neighboring bead $H_4$ or $H_1$.
	In any case $f^{-1}_0(c) + v$ is connected and, hence, the first move makes $f^{-1}_1(c)$ a color component that contains $v$, the pivot of the second move.
	Accordingly, move two recolors all vertices in $f^{-1}_1(c)$ and thereby eliminates $c$.
	Completing the strategy with the remaining $|C|-1$ colors takes linear time.

	It is easy to see that this all works in polynomial time.
\end{proof}

\ThinSpiderNPComplete*
\renewcommand{\bar}[1]{\overline{#1}}
\newcommand{\dbar}[1]{\overline{\overline{#1}}}
\newcommand{\ubar}[1]{\underline{#1}}
\newcommand{\dubar}[1]{\underline{\underline{#1}}}
\newcommand{\oline}[1]{\overline{#1}}
\newcommand{\doline}[1]{\overline{\overline{#1}}}
\newcommand{\uline}[1]{\underline{#1}}
\newcommand{\duline}[1]{\underline{\underline{#1}}}

\begin{proof}
    Since both problems are known members of $\NP$ even on general graphs, we can focus on the $\NP$-hardness.
    The proof works as follows.
    For a given graph $H$ on $n$ vertices and $m$ edges and an integer $\kappa$, we (i) describe a polynomial time construction of a precolored thin spider $G$ and the integer $k=\kappa+|V(H)|$, (ii) show that the existence of a vertex cover $W$ in $H$ with $|W| = \omega$ implies the existence of a \emph{fixed} strategy for $G$ that has $n+\omega$ moves, and (iii) prove that any free strategy for $G$ with at most $\sigma$ moves implies that there is a vertex cover $W$ of $H$ with $|W| = \sigma-n$.
    At the end of the proof, we conclude from the three steps that the two statements in the lemma are implied.

	\medskip
	To begin with Stage~(i), let $H = (V,E)$ be a graph with $n$-vertex set $V$ and $m$-edge set $E$ and let $\kappa \in \N$.
	For technical reasons, we assume that $V$ is an ordered set.
	We construct the following precolored split graph $G = (K, I, A)$ on clique $K$, independent set $I$, edges $A$, and with initial {\flooding} $f_0$.
	Let $\dot{V} = \{\dot{v} \mid v \in V\}$ and $\ddot{V} = \{\ddot{v} \mid v \in V\}$ be two copies of the vertex set $V$ and define the following four sets of ordered pairs to represent the edge set $E$:
	\begin{align*}
		\bar{E} = \{\bar{vw} \mid vw \in E, v < w \},\\
		\ubar{E} = \{\ubar{vw} \mid vw \in E, v < w \},\\
		\dbar{E} = \{\dbar{vw} \mid vw \in E, v < w\},\\
		\dubar{E} = \{\dubar{vw} \mid vw \in E, v < w\}.\\
	\end{align*}
	Note that in the definition above, we use the vertex ordering only to be able to uniquely choose a \emph{first} and \emph{second} element from unordered pairs.
	Based on this, we define that both, clique and independent set, are created from these copies as $K = \dot{V} + \bar{E} + \ubar{E}$ and $I = \ddot{V} + \dbar{E} + \dubar{E}$.
	Additional to the edges that are present between all pairs of distinct clique vertices, we add exactly the following edges between $K$ and $I$:
	For every vertex $v \in V$, the edge set $A$ contains $\dot{v}\ddot{v}$ and, for every edge $e \in E$ of the graph $H$, we include the edges $\bar{e} \, \dbar{e}$ and $\ubar{e}\,\dubar{e}$ in $A$.
	It is easy to see that $G$ is a thin spider with clique $K$ and independent set $I$ where every vertex $x \in K$ sees exactly one vertex $y \in I$ and vice versa.
	
	As the last step, the construction defines the initial {\flooding} $f_0: K + I \rightarrow V + \{0\}$, where $V$ is used as color set and $0$ is an additional color not present in $V$.
	In fact, for all $x \in K+I$, let
	\[f_0(x) = \begin{cases}
		0, &\text{if } x \in \dot{V},\\
		v, &\text{if } x = \oline{vw} \in \bar{E},\\
		w, &\text{if } x = \uline{vw} \in \ubar{E},\\
		v, &\text{if } x \in \ddot{V},\\
		w, &\text{if } x = \doline{vw} \in \dbar{E},\\
		v, &\text{if } x = \duline{vw} \in \dubar{E}.\\
		\end{cases}\]
	The integer $k$ is simply defined as $k = n+\kappa$.
	Apparently, for a given pair $H, \kappa$, the pair $G,k$, including $f_0$, can be computed in polynomial time, that is, due to the number of edges between vertices in $K$, in $\cO(n^2 + m^2)$ time.

	\begin{claim}\label{clm:NP-completeness:thin_spiders:VC_implies_FixedFlooding}
		If the given $n$-vertex graph $H$ has a vertex cover $W$ with $|W| = \omega$ elements then the constructed graph $G$ with initial {\flooding} $f_0$ has a fixed strategy of $n + \omega$ moves for every fixed pivot chosen from $\dot{V}$.
	\end{claim}

	Showing Claim~\ref{clm:NP-completeness:thin_spiders:VC_implies_FixedFlooding} provides Stage~(ii) of the proof.
	Select any vertex $p \in \dot{V}$ as the fixed pivot.
	Then assume a vertex cover $W \subseteq V$ of $H$ with $|W| = \omega$, that is, for every edge $vw \in E$, $v \in W$ or $w \in W$.
	Let $w_1, \dots, w_{\omega}$ be any ordering of $W$ and, similarly, let $v_1, \dots, v_{n-\omega}$ be any ordering of $V-W$.
	We show that the fixed move order
	\[S = w_1, \dots, w_\omega, v_1, \dots, v_{n-\omega}, w_1, \dots, w_\omega\]
	is a fixed strategy for $G$ with pivot $p$ that consists of $n+\omega$ moves.
	
	Since the correct number of moves in $S$ follows from construction, we are left with the proof that $S$ floods $G$ monochromatically.
	The first $\omega$ moves flood with all colors in $W$ and, so, for every $H$-edge $vw \in E$ at least one $G$-vertex of $\oline{vw} \in \bar{E}$ and $\uline{vw} \in \ubar{E}$ has already been conquered from $p$, thus, it is already in $f^{-1}_\omega(p)$.
	In fact, $\oline{vw} \in f^{-1}_\omega(p)$, if $v \in W$, and $\uline{vw} \in f^{-1}_\omega(p)$, if $w \in W$.
	Thus, both of the two vertices $\oline{vw}$ and $\uline{vw}$ have been conquered and included in $f^{-1}_\omega(p)$ only if $v$ and $w$ are both part of the vertex cover.
	
	Using these facts about $f_\omega$, we can conclude that, after move $n$, when every color of $V$ has been used once, we have, for all edges $vw \in E$, that
	\begin{itemize}
		\item both, $\oline{vw}$ and $\uline{vw}$, are in $f^{-1}_n(p)$,
		\item the vertex $\doline{vw} \in \dbar{E}$ is in $f^{-1}_n(p)$, if $v \in W$, and $\duline{vw} \in \dubar{E}$ has been added to $f^{-1}_n(p)$, if $w \in W$,
		\item $\doline{vw} \not\in f^{-1}_n(p)$ implies that $w \in W$ and, similarly, $\duline{vw} \not\in f^{-1}_n(p)$ implies $v \in W$.
	\end{itemize}
	Also, this means that $K \subseteq f^{-1}_{n}(p)$ since $\dot{V} \subseteq f^{-1}_i(p)$ for all $i \in \{0, \dots, n\}$.
	Moreover, since every color of $V$ has been used once and because every vertex $\ddot{v} \in \ddot{V}$ is a neighbor of $\dot{v} \in \dot{V} \subseteq f^{-1}_0(p)$, we find that $\ddot{V} \subseteq f^{-1}_{n}(p)$, too.
	
	Now, if a vertex $x$ of $G$ is not in $f^{-1}_n(p)$ then either $x = \doline{vw}$ with $w \in W$ or $x = \duline{vw}$ with $v \in W$.
	Thus, $x$ is in $I$, the only neighbor of $x$ is in $K \subseteq f^{-1}_n(p)$, and the color $f_n(x)$ is in $W$.
	Since the last $\omega$ moves of $S$ repeat the colors $W$, the region $f^{-1}_{n+\omega}(p)$ equals $K+I$ and, thus, $f_{n+\omega}$ monochromatically floods $G$ with $w_\omega$.
	This completes the proof of Claim~\ref{clm:NP-completeness:thin_spiders:VC_implies_FixedFlooding}.

	\medskip
	The next Stage~(iii) basically proves the other direction, that is, free strategies for $G$ of $\sigma$ moves imply the existence of a vertex cover for $H$ with $\sigma-n$ elements.
	But before we can go into this untroubled, we need the following normalization step.
	\begin{claim}\label{clm:NP-completeness:thin_spiders:FreeFloodingNormalization}
		If there is a free strategy for $G$ with $\sigma$ moves then there is also a free strategy $(x_1,c_1), \dots, (x_\sigma,c_\sigma)$ for $G$ where $\{x_1, \dots, x_\sigma\} \subseteq K$. 
	\end{claim}
	If $S$ is a $\sigma$-move free strategy for $G$ then, for the sake of an induction parameter, $S$ is called $s$-independent for some $s \in \{0, \dots, \sigma\}$, if exactly $s$ moves of $S$ apply pivots in $I$.
	We use induction on $s$ to show the result.
	For the start at $s=0$, we observe that any $0$-independent strategy for $G$ must have all pivots in $K$.

	\smallskip
	In the induction step at $s > 0$, we may assume that the claim holds for all $s'$-independet strategies for $G$ with $s' < s$.
	So, let $S=(x_1,c_1), \dots, (x_\sigma,c_\sigma)$ be any $s$-independent strategy for $G$ and let $i$ be the minimum index of a move $(x_i,c_i)$ with $x_i \in I$.
	Since $G$ is a thin spider, $x_i$ has exactly one neighbor in $K$, say $y$.
	
	In the following, we create a new move order $S' = (x'_1, c'_1), \dots, (x'_\sigma, c'_\sigma)$ from $S$ and let $F_0=f_0$ and $F_1, \dots, F_\sigma$ be the corresponding {\floodings} and, accordingly, for all $x \in K+I$ and all $j \in \{1, \dots, \sigma\}$, we denote the {\floodareas} with $F^{-1}_j(x)$.
	In fact, for all $j \in \{1, \dots, i-1\}$, we simply let $(x'_j, c'_j) = (x_j,c_j)$.
	If $y \in f^{-1}_{i-1}(x_i)$ then $(x'_i, c'_i) = (y, c'_i)$ and, for all $j \in \{i+1, \dots, \sigma\}$, we let $(x'_j, c'_j) = (x_j,c_j)$, again.
	In this case, $f_j = F_j$ for all $j \in \{1, \dots, \sigma\}$ and, thus, $S'$ is a $(s-1)$-independent strategy for $G$.
	By the induction assumption, the claim is true.

	Otherwise, $S$ has not changed the color of $x_i$ until move $i$, and we have $f_{i-1}(x_i) = \dots = f_1(x_i) = f_0(x_i) = v \in V$.
	Here, we may quietly assume that move $i$ is not useless in $S$ and, thus, $c_i = f_{i-1}(y)$ to merge the {\floodareas} of $x_i$ and $y$.
	Then we leave out move $i$ in $S'$ and, instead, let $(x'_j,c_j) = (x_{j+1}, c_{j+1})$ for all $j \in \{i, \dots, \sigma-1\}$.
	The final move $(x'_\sigma,c_\sigma)$ of $S'$ is then defined as $(\dot{v},v)$.
	Again, $F_j = f_j$ for all $j \in \{1, \dots, i-1\}$.
	Because the move $(x_i, c_i)$ changes only the color of the single vertex $x_i$, a vertex that is solely adjacent to $y \in K$, and since $(x'_j,c_j) = (x_{j+1}, c_{j+1})$ for all $j \in \{i, \dots, \sigma-1\}$, it is easy to see that either $F^{-1}_j(x) = f^{-1}_{j+1}(x)$ or $F^{-1}_j(x) = f^{-1}_{j+1}(x) - x_i$, for all $j \in \{i, \dots, \sigma-1\}$ and all $x \in K+I - x_i$.
	Hence, $F^{-1}_{\sigma-1}(\dot{v}) \supseteq f^{-1}_{\sigma}(\dot{v}) - x_i = K+I-x_i$, which means that $F_{\sigma-1}$ is nearly monochromatic except for maybe $x_i$.
	Since $f_0(x_i) = F_0(x_i) = v$, the final move $(\dot{v},v)$ conquers $x_i$ and makes $F_\sigma$ entirely monochromatic.
	Now $S'$ is again a $(s-1)$-independent strategy for $G$ and the claim follows from the induction assumption.
	This completes the proof of Claim~\ref{clm:NP-completeness:thin_spiders:FreeFloodingNormalization}.

	\medskip
	The following claim builds on this normalization to finish Stage~(iii).
	\begin{claim}\label{clm:NP-completeness:thin_spiders:FreeFlooding_implies_VC}
		If the constructed graph $G$ with initial {\flooding} $f_0$ has a free strategy of $\sigma$ moves that all use a pivot in $K$ then the given $n$-vertex graph $H$ has a vertex cover $W$ with $|W| = \sigma-n$ elements.
	\end{claim}
	Let $S = (x_1,c_1), \dots, (x_\sigma,c_\sigma)$ be a free strategy for $G$ with $\{x_1, \dots, x_\sigma\} \subseteq K$.
	Although $S$ is free, the idea is to reinterpret every move from the perspective of a fixed pivot in $\dot{V}$.
	So, in the remainder of the claim proof, fix any vertex $p$ in $\dot{V}$.
	Based on this, we want to reverse the idea of Claim~\ref{clm:NP-completeness:thin_spiders:VC_implies_FixedFlooding} and disregard colors $v\in V$ for $W$ that are conquered in a single move and, instead, use just the remaining colors for the vertex cover.
	
	However, in a free strategy, a move does not necessarily need to apply color $v$ to conquer the respective areas in $G$.
	In order to describe the intended conquered color, we introduce the following move categories in relation to the pseudo pivot $p$.
	More precisely, for a color $v \in V$, a move $(x_i, c_i)$ in $S$ is called 
	\begin{itemize}
		\item a hard $v$-move, if $x_i \in f^{-1}_{i-1}(p)$ and $c_i = v$, or
		\item a soft $v$-move, if $x_i \not\in f^{-1}_{i-1}(p)$ and there is a vertex $x \in f^{-1}_{i-1}(x_i)$ with $f_0(x) = f_{i-1}(x) = v$.
	\end{itemize}
	Notice that the two conditions exclude each other.
	Beyond that, every move $(x_i, c_i)$ in $S$ is a $v$-move for at most one $v \in V$.
	In fact, if $x_i \in f^{-1}_{i-1}(p)$ then $i$ is a hard $v$-move just for $v = c_i$ and being a soft $w$-move for $w \in V-v$ is excluded.
	Otherwise, when $x_i \not\in f^{-1}_{i-1}(p)$, having a hard move $i$ is excluded.
	Under this condition, we see from the definition that the color of vertices in $f^{-1}_{i-1}(x_i)$ are all the same.
	So, if a vertex $x$ with $f_0(x) = f_{i-1}(x) = v$ exist in $f^{-1}_{i-1}(x_i)$, then $v$ is uniquely determined and $i$ is a soft $v$-move for exactly this $v$.
	
	Notice that some moves may neither be a hard nor a soft $v$-move for any $v \in V$.
	We furthermore observe, for every color $v \in V$, that there is at least one hard $v$-move $(x_i, c_i)$ (hence, where $x_i \in f^{-1}_{i-1}(p)$ and $c_i = v$), because this is the only way to conquer $\ddot{v} \in I$ from a pivot in $K$.
	
	\smallskip
	Now, let $W'$ consist of all colors $v \in V$ (which excludes $0$) that provide more than one $v$-move in $S$, that is, where distinct $i, j \in \{1, \dots, \sigma\}$ exist such that $(x_i,c_i)$ and $(x_j,c_j)$ are both $v$-moves, hard or soft.
	As we have just seen, $S$ contains at least one hard $v$-move for every color $v \in V$ and all other elements of $S$ are $v$-moves for at most one $v\in V$.
	So, there can be at most $\sigma-n$ colors $v \in V$ that appear as more than one $v$-move and, hence, $|W'| \leq \sigma - n$.
	To precisely hit the quota, we assume that $W$ is a super set of $W'$ that is filled with $\sigma - n - |W'|$ additional and arbitrary vertices from $V-W'$ so that $|W| = \sigma - n$.
	
	It remains to show that $W$ is a vertex cover of $H$.
	For a contradiction, assume that there is an edge $vw \in E$ such that $v,w \in V-W$.
	By construction, this would mean that the colors $v$ and $w$ appear in only one $v$-move and one $w$-move of $S$, respectively.
	Say $(x_i,c_i)$ is the only $v$-move and $(x_j,c_j)$ the only $w$-move in $S$, where, without loss of generality, $i < j$.
	As explained above, these moves must be hard and, so, we would find that $x_i \in f^{-1}_{i-1}(p)$ and $c_i = v$ and, accordingly, $x_j \in f^{-1}_{j-1}(p)$ and $c_j = w$.
	
	To derive the contradiction, we begin to argue that $f_i(\uline{vw}) = w$ and $f_i(\duline{vw}) = v$.
	First, assume that at least one of the conditions, $f_i(\uline{vw}) \not= w$ and $f_i(\duline{vw}) \not= v$, holds.
	Then we select the earliest move $(x_\ell, c_\ell)$ where $\uline{vw} \in f^{-1}_{\ell-1}(x_\ell)$.
	Clearly, this move exists with $\ell \leq i$.
	In fact, if $f_i(\uline{vw}) \not= w$, this is because $\uline{vw}$ is initially $w$-colored and there must be a prior move that recolors $\uline{vw}$ with a different color $c_\ell$.
	In the other case, we notice that recoloring any vertex in $I$, like $\duline{vw}$, requires conquering it from a pivot in $K$, beforehand.
	This is only possible with a move like $\ell$ that floods the only neighbor $\uline{vw}$ in the color $c_\ell = f_0(\duline{vw}) = v$.

	The vertex $x_\ell$ cannot be in $f^{-1}_{\ell-1}(p)$ because this would mean $\uline{vw} \in f^{-1}_{\ell-1}(x_\ell) = f^{-1}_{\ell-1}(p)$.
	Since $\ell$ is the first move recoloring $\uline{vw}$ this would require the {\floodareas} of $p$ and $\uline{vw}$ being merged in a prior move $\ell'$ with $x_{\ell'} \in f^{-1}_{\ell'}(p)$ and $c_{\ell'} = f_0(\uline{vw}) = w$.
	This is a hard $w$-move with $\ell' < \ell \leq i < j$, which is impossible.
	But having $x_\ell$ not in $f^{-1}_{\ell-1}(p)$ is not working either.
	More precisely, there is $\uline{vw} \in f^{-1}_{\ell-1}(x_\ell)$ with $f_0(\uline{vw}) = f_{\ell-1}(\uline{vw}) = w$.
	This means that $\ell$ is a soft $w$-move with $\ell \leq i < j$, again impossible.

	Hence, after the only $v$-move $(x_i, c_i)$, we still have $f_i(\uline{vw}) = w$ and $f_i(\duline{vw}) = v$.
	Since $x_i \in f^{-1}_{i-1}(p)$ and $c_i = v$, the {\floodarea} $f^{-1}_i(p)$ is $v$-colored.
	However, since the only neighbor $\uline{vw}$ of $\duline{vw}$ is not $v$-colored, it follows that $\uline{vw},\duline{vw} \not\in f^{-1}_i(p)$.
	As discussed above, conquering $\duline{vw}$ from a pivot in $K$ after move $i$ works only via a prior move that recolors $\uline{vw}$.
	Thus, let $(x_\ell,c_\ell)$ be the earliest move with $i < \ell$ that recolors $\uline{vw}$, that is, with $\uline{vw} \in f^{-1}_{\ell-1}(x_\ell)$.
	
	If $x_\ell$ is in $f^{-1}_{\ell-1}(p)$ then $\uline{vw} \in f^{-1}_{\ell-1}(p)$, too.
	In the same way as described above, this means that there is a prior hard $w$-move that conquers $\uline{vw}$ from a pivot in the {\floodarea} of $p$.
	This can only be move $j$.
	But afterward, conquering $\duline{vw}$ from a pivot in the {\floodarea} of $\uline{vw}$ produces another hard $v$-move, which is impossible.
	Otherwise, when $x_\ell \not\in f^{-1}_{\ell-1}(p)$, we again have $\uline{vw} \in f^{-1}_{\ell-1}(x_i)$ with $f_0(\uline{vw}) = f_{\ell-1}(\uline{vw}) = w$.
	This means that $\ell$ is a soft $w$-move, but noticeably not a hard one.
	Together with the hard $w$-move $j$, we would have at least two $w$-moves, a contradiction.
	Hence, $W$ must be a vertex cover of $H$ with $|W| = \sigma-n$, which completes the proof of Claim~\ref{clm:NP-completeness:thin_spiders:FreeFlooding_implies_VC}.
	
	\medskip
	We are finally ready to prove Lemma~\ref{lem:NP-completeness:thin_spiders} via polynomial time reductions of the \textsc{Minimum Vertex Cover} problem to {\MinimumFloodIt} and {\MinimumFreeFloodIt}, respectively.
	In both cases, the reduction is implemented by the computation of $G, k$ and $f_0$ from a given graph $H$ with integer $\kappa$, which is in polynomial time, as explained above.
	
	If $H$ has a vertex cover $W$ with $|W|\leq \kappa$ then, according to Claim~\ref{clm:NP-completeness:thin_spiders:VC_implies_FixedFlooding} $G$ has a fixed strategy of $n+|W| \leq n+\kappa = k$ moves.
	Since fixed flooding is a special case of free flooding, Claim~\ref{clm:NP-completeness:thin_spiders:VC_implies_FixedFlooding} establishes the first direction for both flooding problems.

	Reversely, if $G$ has a strategy of at most $k$ moves, free or fixed with pivot in $\dot{V}$, then Claim~\ref{clm:NP-completeness:thin_spiders:FreeFloodingNormalization} secures the existence of a free strategy $S$ for $G$ that has the same number $|S|$ of moves and where all pivots are chosen in $K$.
	Via Claim~\ref{clm:NP-completeness:thin_spiders:FreeFlooding_implies_VC}, this implies the existence of a vertex cover $W$ for $H$ with $|W| = |S|-n \leq k-n = \kappa$ elements.
	This completes the proof.
\end{proof}


\section*{Disclaimer on the Use of AI}

We note that the AI system \textsc{ChatGPT}, in particular the models \emph{GPT-5 mini} and \emph{GPT-5.3}, has been used in this work for the only purpose of improving the presentation and readability of the text, as the authors are not native English speakers.
No other use of \textsc{ChatGPT} or its subroutines has been made in this paper.
In particular, no AI framework has contributed to any creative decisions, the development of original results, or the construction of rigorous mathematical proofs, all of which were carried out entirely by the authors.


\bibliographystyle{abbrv}
\bibliography{floodit}

\end{document}